\documentclass[journal,onecolumn,11pt,draftcls,doublespace]{IEEEtran}

\IEEEoverridecommandlockouts
\ifCLASSINFOpdf
\else
\fi
\usepackage[cmex10]{amsmath}
\usepackage{amsthm}
\usepackage{amssymb}
\usepackage{mathrsfs}
\usepackage{amsbsy}
\usepackage{mathrsfs}
\usepackage{setspace}
\usepackage{graphicx}
\usepackage{caption}
\usepackage{subcaption}
\usepackage{psfrag}
\usepackage{newclude}
\usepackage{latexsym}
\usepackage[table]{xcolor}
\usepackage{multirow}
\usepackage{rotating}
\usepackage{booktabs}
\usepackage{amsmath}
\usepackage{mathrsfs}
\usepackage{bbm} % indicator function
\usepackage{filecontents}
\usepackage{amsmath,epsfig,cite,amsfonts,amssymb,psfrag,subfig}

\def\A{{A^n_\epsilon}}

\def\bs{{\boldsymbol{s}}}

\def\bu{{\boldsymbol{u}}}
\def\bv{{\boldsymbol{v}}}

\def\bx{{\boldsymbol{x}}}

\def\by{{\boldsymbol{y}}}

\def\pr{{\text{P}}}
\def\c{{\mathcal{C}}}
\usepackage{accents}
\newlength{\dhatheight}
\newcommand{\doublehat}[1]{%
    \settoheight{\dhatheight}{\ensuremath{\hat{#1}}}%
    \addtolength{\dhatheight}{-0.35ex}%
    \hat{\vphantom{\rule{1pt}{\dhatheight}}%
    \smash{\hat{#1}}}}
\newtheorem{remark}{Remark}
\newtheorem{theorem}{Theorem}

\newtheorem{proposition}{Proposition}

\graphicspath{{figures/}}
\allowdisplaybreaks

\begin{document}
\title{Multi-layer Gelfand--Pinsker Strategies for the Generalized
  Multiple-Access Channel} \author{\IEEEauthorblockN{\normalsize
    Mohammad Javad Emadi$^*$, Majid Nasiri Khormuji$^\dag$, Mikael
    Skoglund$^\ddagger$ and Mohammad Reza Aref$^\star$}
  \IEEEauthorblockA{\small $^*$EE Dept., Amirkabir University of Technology, Tehran, Iran\\
    $^\dag$Huawei Technologies Sweden,\\
    $^\ddagger$Royal Institute of Technology (KTH), EE Dept. and the ACCESS Linnaeus Center, Stockholm, Sweden\\
    $^\star$Information Systems and Security Lab (ISSL), EE Dept., Sharif University of Technology, Tehran, Iran\\
    E-mail(s): mj.emadi@aut.ac.ir, majid.nk@huawei.com,
    skoglund@ee.kth.se, aref@sharif.edu}

\thanks{This work was supported in part by the Iranian National Science Foundation (INSF) under Contract No. 88114.46-2010
and by the Iran Telecom Research Center (ITRC) under Contract No. 500.18495.}
 }
\maketitle
%%%%%%%%%%%%%%%%%%%%%%%%%%%%%%%%%%%%%%%%%%%%%%%%%%%%%%%%%%%%%%%%%%%%%%%%%%%%%%%%%%%%%%%%%%%%%%%%%%%%%%%%%%%%%%%%%%%%%%%%%%%%%%%%%%%%
\begin{abstract}
  We study a two-user state-dependent generalized multiple-access channel (GMAC) with correlated states. It is assumed that each encoder has \emph{noncausal} access to channel state information (CSI). We develop an achievable rate region by employing rate-splitting, block Markov encoding, Gelfand--Pinsker multicoding, superposition coding and joint typicality decoding.  In the proposed scheme, the encoders use a partial decoding strategy to collaborate in the next block, and the receiver uses a backward decoding strategy with joint unique decoding at each stage. Our achievable rate region includes several previously known regions proposed in the literature for different scenarios of multiple-access and relay channels. Then, we consider two Gaussian GMACs with additive interference. In the first model, we assume that the interference is known noncausally at both of the encoders and construct a multi-layer Costa precoding scheme that removes \emph{completely} the effect of the interference. In the second model, we consider a doubly dirty Gaussian GMAC in which each of interferences is known noncausally only at one encoder.  We derive an inner bound and analyze the achievable rate region for the latter model and interestingly prove that if one of the encoders knows the full CSI, there exists an achievable rate region which is \emph{independent} of the power of interference.
%  Then, we consider a Gaussian GMAC with an additive interference which is known noncausally at both of the encoders and construct a   \emph{multi-layer} Costa precoding scheme that removes completely the effect of the interference. Afterwards, an achievable rate region is derived for a doubly dirty Gaussian GMAC in which each of interferences is known noncausally only at one encoder. We conjecture that the rate region is a convex combination four simple rate regions. Finally, we prove that if one of the encoders knows the full CSI, there exists an achievable rate region which is independent of the power of interference.
\end{abstract}
\IEEEpeerreviewmaketitle
%\vspace{-0.6cm}
%%%%%%%%%%%%%%%%%%%%%%%%%%%%%%%%%%%%%%%%%%%%%%%%%%%%%%%%%%%%%%%%%%%%%%%%%%%%%%%%%%%%%%%%%%%%%%%%%%%%%%%%%%%%%%%%%%%%%%%%%%%%%%%%%%%%
\section{Introduction}
Achieving higher throughput, reliability and robustness against channel variations is an important requirement for next generation wireless networks. Toward this end, spatial diversity, which is possible by use of multiple transmit antennas, is an attractive solution. However, due to restricted size of a mobile end-devices, this technique can be impractical for uplink communication scenarios. Thus, cooperation between nodes is introduced to achieve the diversity gain. In wireless communication networks there is freedom to establish cooperation between users, since while one user transmits its signal, other users in its neighborhood can overhear a noisy version of the transmission (which can be modelled as a type of feedback signal). Hence, the neighbors can play the role of relay nodes to help the original user to convey its message to the receiver. Imagine an uplink channel with two users and one receiver such that in addition to the receiver, each user can overhear each others' transmissions. We refer to this channel as a GMAC which is also known as the MAC with generalized feedback. Carleial in \cite{carleial1982multiple} derived an achievable rate region for the discrete memoryless (DM) GMAC. Willems in \cite{willemsThesis} studied the DM GMAC and proposed an achievable rate region based on the backward decoding technique. In general, it is difficult to compare these two achievable rate regions. In \cite{Zeng1989IT}, it is proved that for some special cases, Willems' achievable rate region includes Carleial's. Willems in \cite{Willems85} has also established the capacity region of DM MAC with cribbing encoders as an special case of DM GMAC wherein each user \emph{ideally} receives the transmitted signal of the other user with one channel use delay. We highlight that the DM GMAC also includes the classical relay channel \cite{cover79} as a special case.

On the other hand, one important class of channels is formed by those where the statistics of the channel are controlled by random parameters. These channels are referred to as state-dependent channels which have a wide variety of applications in wireless communications, writing on memories with defects, information embedding (IE), steganography  and authentication check. In particular, fading parameters and interference signals can be interpreted as CSI in a wireless communication system. Depending on the network configuration, the CSI may be available at some nodes in the network. Shannon in \cite{shannon1958channels} established the capacity of a point-to-point DM state-dependent channel in which the CSI is \emph{causally} available at the encoder. Later the capacity of the channel with \emph{noncausal} CSI at the transmitter (CSIT) was established by Gelfand--Pinsker coding (GPC) technique in \cite{gel1980coding}. Costa in \cite{costa1983writing} considered the Gaussian point-to-point channel with additive interference wherein the interference (i.e. the channel state) was known noncausally at the encoder. Remarkably, Costa proved that the capacity of the Gaussian channel with noncausal knowledge of the interference at the encoder is equal to that with no interference. Interestingly, IE can be modeled as a problem of state-dependent channel with noncausal CSIT, where the host signal plays the role of side information. Moreover, problem of distributed IE (i.e., multiuser IE) is studied in the literature \cite{kotagiri2005reversible,zaidi2007broadcast}. For a literature survey on state-dependent channels and its applications see \cite{keshet2007channel}. Besides, problems of imperfect CSI at the transmitter(s) and/or receiver(s) are investigated in \cite{pablothesis2007}.

The state-dependent MAC subject to different conditions has been considered in the literature and achievable rate regions have been established. For instance based on the availability of noncausal CSI at some encoders, an achievable rate region is derived in \cite{somekh2008cooperative,kotagiri2008multiaccess,KhosraviITW2011,EmadiITW2012}. Since there is a connection between receiving feedback signals at the transmitters and causal knowledge about the CSI, state-dependent MAC channel models with causal CSIT are studied in \cite{lapidoth2010multiple,li2010multiple,zaidi2012capacity}. In particular, the state-dependent DM MAC with correlated states and ideal cribbing encoders is investigated in \cite{EmadiIET2012} and bounds on the capacity region are derived, which reduced to the capacity region of the \emph{asymmetric} state-dependent DM MAC with noncausal CSI available at an ideal cribbing encoder \cite{Bross10}.

State-dependent relay channels are also studied based on knowledge of CSI at the source and/or at the relay. In \cite{KhosraviITW2011} an achievable rate is derived for a case that the source knows the full CSI noncausally and the relay knows the CSI partially.  In \cite{abdellatif2009lower,zaidi2010cooperative,ZaidiIT2012,nasiri2013state} lower and upper bounds are derived for the cases of availability of CSI at the source or at the relay.

In this paper we study a two-user state-dependent DM GMAC with correlated states such that each encoder has access noncausally to \emph{partial} CSI. Also motivated by the broadcast nature of the wireless transmission medium, it is assumed that both of encoders receive feedback signals from the channel output and then try to cooperate with each other by use of the feedback signals. A motivating scenario to investigate such a model is the uplink in a wireless cellular system, where some of the users may be capable of sensing the transmissions of other users over the network. These users are sometimes referred to as cognitive users. The unwanted transmitted signals can be modeled as known interference signals (state information) to the cognitive users. Hence, the cognitive users can cooperatively send their information to a destination in order to cope with interfering signals. This model is also useful in implementing user cooperation to realize \emph{virtual} multiple-input multiple-output (MIMO) systems where the conventional collocated MIMO transmission is not feasible. As another example, it is well known that the performance of Long-Term Evolution (LTE) systems can be improved by use of CSI at the transmitters. Moreover by use of the cooperative strategy at the uplink of the LTE, throughput and reliability will be improved thanks to the diversity gain offered by user cooperation. The state-dependent DM GMAC with correlated states has application also in modeling ad-hoc networks. In the ad-hoc network, the nodes form a network based on the CSIT. Then, two nodes with good intermediate link cooperate to send their information to a receiver.

Considering the state-dependent GMAC, we present an achievable rate region using rate-splitting, block-Markov encoding (BME), GPC, superposition encoding, \emph{partial decode-and-forward} at the encoders and \emph{backward decoding} at the receiver. Our proposed achievable rate region includes several known results as special cases reported in \cite{willemsThesis,Willems85,cover79,KhosraviITW2011,EmadiIET2012,Bross10,abdellatif2009lower,zaidi2010cooperative,ElgamalAref82,PhilosofIT09}. In particular, for a Gaussian model, we prove that if full CSI is available at both of the encoders, the effect of the interference is \emph{completely} removed by using \emph{multi-layer} Costa precoding. Moreover, we prove that if at least one of the users knows the full CSI, there exists an achievable rate region which is independent of the interference power. Our results shed further light on fundamental limits of cooperative strategies between users in a state-dependent wireless
channels, which can be used as guidelines in implementation of future cooperative communication systems.
%The paper provides better understanding of cooperation strategies between users in a state-dependent wireless channel, which guiding the design and implementation of future cooperative communication systems.

\emph{Organization:} The remaining part of the paper is organized as follows. The state-dependent DM GMAC is introduced in Section II. The proposed achievable rate region along with special cases, encoding and decoding strategies are discussed in Section III. The Gaussian GMAC with additive interferences that are partially known at the encoders is studied in Section IV. The paper is concluded in Section V.

\emph{Notations:} We use uppercase and lowercase letters to denote random variables and their realizations, respectively. The probability of an event $\mathcal{A}$ is denoted by $\pr(\mathcal{A})$ and the conditional probability of $\mathcal{A}$ given $\mathcal{B}$ is denoted by $\pr(\mathcal{A}|\mathcal{B})$, and  $p_{Y|X}(y|x)$ denotes a collection of conditional probability mass functions (pmfs) on $Y$, one for every $x$. We use boldface letters to denote a vector of length $n$; e.g., $\bx=(x_1,x_2,...,x_n)$ and $\bx_i^j=(x_i,x_{i+1},...,x_j)$ for $i\leq j$.
%%%%%%%%%%%%%%%%%%%%%%%%%%%%%%%%%%%%%%%%%%%%%%%%%%%%%%%%%%%%%%%%%%%%%%%%%%%%%%%%%%%%%%%%%%%%%%%%%%%%%%%%%%%%%%%%%%%%%%%%%%%%%%%%%%%%
\label{channelmodel}
\begin{flushleft}
\begin{figure}[ut]
\centering
\psfrag{M1}[][][.85]{$m_1$}
\psfrag{M2}[][][.85]{$m_2$}
\psfrag{S1}[][][.85]{$(\bs_0,\bs_1)$}
\psfrag{S2}[][][.85]{$(\bs_0,\bs_2)$}
\psfrag{s}[][][.85]{$(\bs_0,\bs_1,\bs_2)$}
\psfrag{x1}[][][.85]{$\bx_1$}
\psfrag{x2}[][][.85]{$\bx_2$}
\psfrag{y1}[][][.85]{$\by_1$}
\psfrag{y2}[][][.85]{$\by_2$}
\psfrag{y3}[][][.85]{$\by_3$}
\psfrag{M12}[][][.85]{$(\hat{m}_1,\hat{m}_2)$}
\psfrag{P}[][][.87]{$\prod_{t=1}^{n} p(y_{1,t},y_{2,t},y_{3,t}|x_{1,t},x_{2,t},s_{0,t},s_{1,t},s_{2,t})$}
\psfrag{ps}[][][.87]{$\prod_{t=1}^{n}p(s_{0,t})p(s_{1,t}|s_{0,t})p(s_{2,t}|s_{0,t})$}
\psfrag{Encoder 1}[][][0.9]{Encoder 1}
\psfrag{Encoder 2}[][][.9]{Encoder 2}
\psfrag{Decoder}[][][.9]{Decoder}
\includegraphics[scale=0.32]{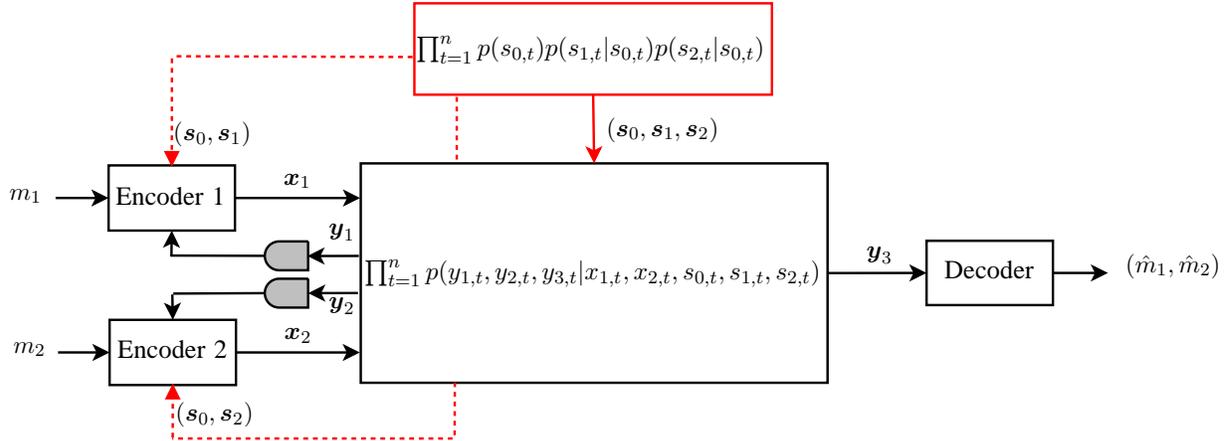}
\caption{\small The two-user state-dependent GMAC. The encoders receive a feedback from the channel with one channel use delay, i.e., strictly causal feedback, and know the partial channel state noncausally.}
\label{fig:cnlmdl}
\vspace{-.7cm}
\end{figure}
\end{flushleft}
%%%%%%%%%%%%%%%%%%%%%%%%%%%%%%%%%%%%%%%%%%%%%%%%%%%%%    Achievable Rate Region
\vspace{-1cm}
\section{Channel Model}
The state-dependent GMAC which is depicted in Fig. 1, is denoted by the triple
\begin{equation*}
(\mathcal{X}_1\times\mathcal{X}_2\times\mathcal{S}_0\times\mathcal{S}_1\times\mathcal{S}_2,p_{Y_1Y_2Y_3|X_1X_2S_0S_1S_2}(y_1,y_2,y_3|x_1,x_2,s_0,s_1,s_2),\mathcal{Y}_1\times\mathcal{Y}_2\times\mathcal{Y}_3)
\end{equation*}
where $x_k \in \mathcal{X}_k$ and $y_k \in \mathcal{Y}_k$, for $k\in\{1,2\}$, denote the transmitted and received symbols at the $k$th encoder, respectively; $s_0 \in \mathcal{S}_0$ and $s_k \in \mathcal{S}_k$ denote the partial CSI at the $k$th encoder; and $y_3 \in \mathcal{Y}_3$ denotes the received symbol at the decoder. We assume that the statistics of the channel
are controlled by random parameters $s_0,~s_1$ and $s_2$, whose partial
knowledge is \emph{noncausally} available at the encoders. That is,
encoder $k$ knows \emph{noncausally} the state-sequence pair $(\bs_0,\bs_k)$. The
interaction among the channel's inputs and outputs of the
channel, over $n$ channel uses, is governed by the
conditional pmf \setlength{\arraycolsep}{0.0em}
\begin{align}
\label{eq: chnlmodel}
\nonumber
p(\by_1,\by_2,\by_3|\bx_1,\bx_2,\bs_0,\bs_1,\bs_2)=\prod_{t=1}^n p_{Y_1Y_2Y_3|X_1X_2S_0S_1S_2}(y_{1,t},y_{2,t},y_{3,t}|x_{1,t},x_{2,t},s_{0,t},s_{1,t},s_{2,t}).
\end{align}
\setlength{\arraycolsep}{5pt}
The $n$-extension of pmf of the channel state is as follows
\begin{equation}
\nonumber
p(\bs_0,\bs_1,\bs_2)=\prod_{t=1}^n p_{S_0}(s_{0,t})p_{S_1|S_0}(s_{1,t}|s_{0,t})p_{S_2|S_0}(s_{2,t}|s_{0,t}).
\end{equation}

A length-$n$ code $\mathcal{C}^n(R_{1},R_{2})$ for the state-dependent DM GMAC (where $R_1$ and $R_2$ are reliable transmission rate of user 1 and 2, respectively) consists of
%A $(2^{nR_1},2^{nR_2},n)$ code for the state-dependent DM GMAC consists of
\begin{itemize}
  \item two independent message sets $\mathcal{M}_k=[1:2^{nR_k}]$, for $k\in\{1,2\}$;
  \item encoding functions ${\varphi_{k,t} (.)}$ for $t\in[1:n] \text{ and } k \in\{1,2\}$ which are defined as
\setlength{\arraycolsep}{0.0em}
\begin{eqnarray*}
\varphi_{k,t}&:&~\mathcal{M}_k \times \mathcal{S}_0^n \times \mathcal{S}_k^n \times \mathcal{Y}_k^{t-1} \longrightarrow \mathcal{X}_k,\\
x_{k,t}&=&\varphi_{k,t}(m_k,\bs_0,\bs_1,\by_{k,1}^{t-1})~:\by_{k,1}^{t-1}=(y_{k,1},y_{k,2},...,y_{k,t-1});
\end{eqnarray*}
  \item a decoding function $\psi$ at the receiver to estimate the transmitted messages as follows
  \begin{eqnarray*}
\psi:~\mathcal{Y}_3^n \longrightarrow \mathcal{M}_1 \times \mathcal{M}_2~:~~~~~(\hat{m}_1,\hat{m}_2)=\psi(\by_3).
\end{eqnarray*}
\end{itemize}
The average error probability of a given code $\mathcal{C}^n(R_{1},R_{2})$ is defined as
\begin{eqnarray*}
&&\pr_e^{(n)}=\frac{1}{2^{n(R_1+R_2)}} \times \sum_{(m_1,m_2)\in \mathcal{M}_1\times \mathcal{M}_2}\pr\big\{\psi(\by_3)\neq(m_1,m_2)\big| m_1 \text{ and } m_2 \text{ are sent}\big\}.
\end{eqnarray*}
A rate pair $(R_1,R_2)$ is said to be achievable for the DM state-dependent GMAC if there exists a code sequence $\mathcal{C}^n (R_1,R_2)$ such that $\pr_e^{(n)}\rightarrow0$ as $n\rightarrow\infty$.
\section{Achievable Rate Region}
\label{achivblrate}
We next present the achievable rate region corresponding to our
proposed scheme.  Our protocol is constructed using rate splitting
(i.e. multi-layer encoding), block Markov encoding and
Gelfand--Pinsker multicoding. Each user via the feedback link
\emph{partially} decodes the transmitted message of the other user and
collaborates in the next block via a block Markov strategy and
superposition coding. The users employ Gelfand--Pinsker coding to
utilize their partial channel state information.
\begin{theorem}\label{th:theorem1}
\emph{For the two-user state-dependent DM GMAC with correlated states, where $(S_0,S_k)$ is \emph{noncausally} known at encoder $k$ for $k\in \{1,2\}$,
the closure of the convex hull of the set
\begin{align}\label{eq: theorem1}\nonumber
\mathcal{R}:=\Big\{\big(R_1,R_2\big):& R_1=R_{12}+R_{13},~R_2=R_{21}+R_{23},\\\nonumber
 &0 \leq  R_{12} < I(V_1;Y_2|S_0S_2UX_2)-I(V_1;S_1|S_0U),\\\nonumber
 &0 \leq  R_{21} < I(V_2;Y_1|S_0S_1UX_1)-I(V_2;S_2|S_0U),\\\nonumber
 &0 \leq  R_{13} < I(V_{13};Y_3|UV_1V_2V_{23})-I(V_{13};S_0S_1|UV_1V_2V_{23})+\delta_1^-,\\\nonumber
 &0 \leq  R_{23} < I(V_{23};Y_3|UV_1V_2V_{13})-I(V_{23};S_0S_2|UV_1V_2V_{13})+\delta_2^-,\\\nonumber
 &R_{13}+R_{23}  < I(V_{13}V_{23};Y_3|UV_1V_2)-I(V_{13}V_{23};S_0S_1S_2|UV_1V_2)+\Delta^-,\\
 &R_{12}+R_{13}+R_{21}+R_{23} < I(UV_1V_2V_{13}V_{23};Y_3)-I(UV_1V_2V_{13}V_{23};S_0S_1S_2)\Big\}~~~~
\end{align}
is achievable for a pmf of the form given in (\ref{eq: pmf}). Since the rate-splitting technique is used, for $k \in \{1,2\}$, $R_{k,3-k}$ and $R_{k3}$ denote the splitted rates of $k$th user's message. The auxiliary random variables $U, V_k$, and $V_{k3}$, are with finite alphabets $\mathcal{U},\mathcal{V}_k$, and $\mathcal{V}_{k3}$, respectively, with the following bounded cardinalities;
\begin{align*}
 |\mathcal{U}|& \leq |\mathcal{S}_0||\mathcal{S}_1||\mathcal{S}_2||\mathcal{X}_1||\mathcal{X}_2|+8,\\
 |\mathcal{V}_k|& \leq |\mathcal{U}||\mathcal{S}_0||\mathcal{S}_1||\mathcal{S}_2||\mathcal{X}_1||\mathcal{X}_2|+8,\\
 |\mathcal{V}_{k,3}|& \leq |\mathcal{U}||\mathcal{V}_k||\mathcal{S}_0||\mathcal{S}_1||\mathcal{S}_2||\mathcal{X}_1||\mathcal{X}_2|+4.
\end{align*}
Also, parameters $\Delta^-$, $\delta_k^-$ for $k \in \{1,2\}$ and the pmf are given as follows}
\begin{align*}
\Delta^- &:= \min\big\{0,\Delta_1,\Delta_2,\Delta_3\big\},\\
  \Delta_k &:=  I(V_k;Y_3|UV_{3-k})-I(V_k;S_0S_k|UV_{3-k}),~~\Delta_3 :=  I(V_1V_2;Y_3|U)-I(V_1V_2;S_0S_1S_2|U),\\\nonumber
 \delta_k^- &:=  \min \big\{0,I(V_k;Y_3|UV_{3-k}V_{3-k,3})-I(V_k;S_0S_k|UV_{3-k}V_{3-k,3})\big\},
% \delta_2^- &:=  \min \big\{0,I(V_2;Y_3|UV_1V_{13})-I(V_2;S_0S_2|UV_1V_{13})\big\}.
\end{align*}
\begin{eqnarray}\label{eq: pmf}\nonumber
&&p_{S_0S_1S_2UV_1V_2V_{13}V_{23}X_1X_2Y_1Y_2Y_3} =\\
&& p_{S_0}p_{S_1|S_0}p_{S_2|S_0}p_{U|S_0} p_{V_1|S_0S_1U}p_{V_2|S_0S_2U}
 p_{V_{13}X_1|UV_1S_0S_1}p_{V_{23}X_2|UV_2S_0S_2}p_{Y_1Y_2Y_3|X_1X_2S_0S_1S_2}.
%&& P_{S_0S_1S_2UV_1V_2V_{13}V_{23}X_1X_2Y_1Y_2Y_3} (s_0,s_1,s_2,u,v_1,v_2,v_{13},v_{23},x_1,x_2,y_1,y_2,y_3)=\\\nonumber
%&&  p_{S_0}(s_0)p_{S_1|S_0}(s_1|s_0)p_{S_2|S_0}(s_2|s_0) p_{U|S_0}(u|s_0) p_{V_1|S_0S_1U}(v_1|s_0s_1u)p_{V_2|S_0S_2U}(v_2|s_0s_2u)\times\\\nonumber
%&& p_{V_{13}X_1|UV_1S_0S_1}(v_{13}x_1|uv_1s_0s_1)p_{V_{23}X_2|UV_2S_0S_2}(v_{23}x_2|uv_2s_0s_2)P_{Y_1Y_2Y_3|X_1X_2S_0S_1S_2}(y_1y_2y_3|x_1x_2s_0s_1s_2).\\
\end{eqnarray}
\end{theorem}
\begin{remark}\emph{The proposed rate region in Theorem \ref{th:theorem1} includes several previously known special cases as summarized in Table~\ref{tabel: specialcases}. For each case, the pmf is specialized and auxiliary random variables with respect to the channel model in Fig.~\ref{fig:cnlmdl} and Theorem \ref{th:theorem1} are given.}
\end{remark}
%%%%%%%%%%%%%%%%%%%%%%%%%%%%%%%%%%%%%%%%%%%%%%%%%%%%%%%%%%%%%%%%%%%%%%%%%%%%%%%%%%%%%%%%%%%%%%%%%%%%%%%%%%%%%%%%%%%%%%%%%%%%%%%%%%%%%%%%%%%%%%%%%%%%%%
%%%%%%%%%%%%%%%%%%%%%%%%%%%%%%%%%%%%%%%%%%%%%%%%%%%   Table of Special cases
\begin{table*}[!ht]\footnotesize
\renewcommand{\arraystretch}{1.5}
\caption{\small Special cases of the proposed achievable rate region in Theorem \ref{th:theorem1}.}\label{tabel: specialcases}
\begin{center}
\begin{tabular}{cc|c|c|l|}\cline{2-4}
\multicolumn{1}{ c| }{} & \cellcolor [gray]{.9} Channel model & \cellcolor [gray]{.9} pmf & \cellcolor [gray]{.9} Auxiliary random variables \\ \cline{1-4}

\multicolumn{1}{ |c| }{}
& GMAC \cite[Theorem 7.1]{willemsThesis} & $p_{Y_1Y_2Y_3|X_1X_2}$ & $V_{13}=X_1$, $V_{23}=X_2$  \\ \cline{2-4}

\multicolumn{1}{ |c| }{ \begin{tabular}{@{}c@{}} DM channels \\ without states: \\$\mathcal{S}_0=\mathcal{S}_1=\mathcal{S}_2=\emptyset$ \end{tabular} }

& \begin{tabular}{@{}c@{}} MAC with ideal strictly \\causal cribbing encoders \\ \cite[Theorem 5]{Willems85}\end{tabular} &
\begin{tabular}{@{}c@{}} $p_{Y_3|X_1X_2},$\\ $_{Y_1=X_2, Y_2=X_1}$ \end{tabular}&
\begin{tabular}{@{}c@{}}$\mathcal{V}_{13}=\mathcal{V}_{23}=\emptyset,$ \\ $V_1=X_1, V_2=X_2$\end{tabular}  \\ \cline{2-4}

\multicolumn{1}{ |c| }{}
&  \begin{tabular}{@{}c@{}} Relaying via partial \\ decode-and-forward \cite{cover79,ElgamalAref82}\end{tabular}&
\begin{tabular}{@{}c@{}} $p_{Y_2Y_3|X_1X_2},$ \\ $_{\mathcal{Y}_1=\emptyset}$ \end{tabular}&
\begin{tabular}{@{}c@{}} $\mathcal{M}_2=\mathcal{V}_2=\mathcal{V}_{23}=\emptyset$ \\ $U=X_2,V_{13}=X_1$ \end{tabular}    \\ \hline\hline

%%%%%%%%%%%%%%%%%%%%%%%%%%%%%%%%%%%%%%%%%%%%%%%%%%%%%%%%%%%%%%%%%%%%%%%%%%%%%%%%%%%%%%%%%%%%%%%%%%%%%%%%%%%%%%%%%%%%%%%%%
\multicolumn{1}{ |c| }{}&
\begin{tabular}{@{}c@{}} Partial CSIT at\\ encoders \cite[Equation (30)]{PhilosofIT09} \end{tabular} &
\begin{tabular}{@{}c@{}} $p_{Y_3|X_1X_2S_0S_1S_2},$\\ $_{\mathcal{Y}_1=\mathcal{Y}_2=\emptyset}$ \end{tabular}&
$\mathcal{U}=\mathcal{V}_1=\mathcal{V}_2=\emptyset$   \\ \cline{2-4}

\multicolumn{1}{ |c| }{\begin{tabular}{@{}c@{}} DM state--dependent\\MAC channels \end{tabular}} &
\begin{tabular}{@{}c@{}} One ideal cribbing\\ and informed encoder \cite{Bross10} \end{tabular} &
\begin{tabular}{@{}c@{}} $p_{Y_3|X_1X_2S_2},$ \\ $_{\mathcal{S}_0=\mathcal{S}_1=\mathcal{Y}_1=\emptyset, Y_2=X_1}$ \end{tabular}&
\begin{tabular}{@{}c@{}} $\mathcal{V}_2=\mathcal{V}_{13}=\emptyset$ \\ $V_1=X_1$\end{tabular} \\ \cline{2-4}

\multicolumn{1}{ |c| }{}  &
\begin{tabular}{@{}c@{}} Ideal cribbing encoders\\ and partial CSIT \cite[Theorem 2]{EmadiIET2012} \end{tabular} &
\begin{tabular}{@{}c@{}} $p_{Y_3|X_1X_2S_0S_1S_2},$ \\ $_{Y_1=X_2, Y_2=X_1}$ \end{tabular}&
\begin{tabular}{@{}c@{}} $\mathcal{V}_{13}=\mathcal{V}_{23}=\emptyset$  \end{tabular} \\ \hline\hline

%\multicolumn{1}{ |c| }{} &
%\begin{tabular}{@{}c@{}} Orthogonal noisy cribbing\\ encoders and partial CSIT \cite{EmadiSwe2012} \end{tabular} &
%\begin{tabular}{@{}c@{}} $p_{Y_3|X_1X_2S_1S_2}\times$ \\ $p_{Y_1|X_2}p_{Y_2|X_1}$ \end{tabular} &
%$S_0=\emptyset$ \\ \hline\hline
%%%%%%%%%%%%%%%%%%%%%%%%%%%%%%%%%%%%%%%%%%%%%%%%%%%%%%%%%%%%%%%%%%%%%%%%%%%%%%%%%%%%%%%%%%%%%%%%%%%%%%%%%%%%%%%%%%%%%%%%%

\multicolumn{1}{ |c| }{}&
\begin{tabular}{@{}c@{}} CSIT at the\\ source \cite[Theorem 1]{abdellatif2009lower} \end{tabular} &
\begin{tabular}{@{}c@{}} $p_{Y_2Y_3|X_1X_2S_1}$,\\ $_{\mathcal{S}_0=\mathcal{S}_2=\emptyset}$ \end{tabular} &
\begin{tabular}{@{}c@{}} $\mathcal{V}_2=\mathcal{V}_{23}=\emptyset$ \\ $U=X_2$ \end{tabular} \\ \cline{2-4}

\multicolumn{1}{ |c| }{\begin{tabular}{@{}c@{}} DM state-dependent\\ relay channel:\end{tabular}} &
\begin{tabular}{@{}c@{}} CSIT at the\\  relay \cite[Theorem 1]{zaidi2010cooperative} \end{tabular} &
\begin{tabular}{@{}c@{}} $p_{Y_2Y_3|X_1X_2S_2},$\\ $_{\mathcal{S}_0=\mathcal{S}_1=\emptyset}$ \end{tabular} &
\begin{tabular}{@{}c@{}} $\mathcal{V}_{13}=\mathcal{V}_2=\emptyset$ \\ $V_1=X_1$ \end{tabular} \\ \cline{2-4}

\multicolumn{1}{ |c| }{\begin{tabular}{@{}c@{}} encoder 2 is a relay,\\ i.e., $\mathcal{M}_2=\mathcal{Y}_1=\emptyset$\end{tabular}}  &
\begin{tabular}{@{}c@{}} degraded CSITs at the\\ source and relay \cite[Theorem 7]{KhosraviITW2011} \end{tabular} &
\begin{tabular}{@{}c@{}} $p_{Y_2Y_3|X_1X_2S_0S_1},$\\$_{\mathcal{S}_2=\emptyset}$ \end{tabular} &
\begin{tabular}{@{}c@{}} $\mathcal{V}_2=\mathcal{V}_{13}=\mathcal{V}_{23}=\emptyset$ \end{tabular} \\ \cline{2-4}

\multicolumn{1}{ |c| }{} &
\begin{tabular}{@{}c@{}} Full CSIT at the\\ source and relay  \cite[Proposition 4]{KhosraviITW2011} \end{tabular} &
\begin{tabular}{@{}c@{}} $p_{Y_2Y_3|X_1X_2S_0},$ \\$_{\mathcal{S}_1=\mathcal{S}_2=\emptyset}$ \end{tabular} &
\begin{tabular}{@{}c@{}} $\mathcal{V}_2=\mathcal{V}_{23}=\emptyset$ \end{tabular} \\ \cline{1-4}
\end{tabular}
\end{center}
\vspace{-.8cm}
\end{table*}
%%%%%%%%%%%%%%%%%%%%%%%%%%%%%%%%%%%%%%%%%%%%%%%%%%%%%%%%%%%%%%%%%%%%%%%%%%%%%%%%%%%%%%%%%%%%%%%%%%%%%%%%%%%%%%%%%%%%%%%%%%%%%%%%%%%%%%%%%%%%%%%%%%%%%
%%%%%%%%%%%%%%%%%%%%%%%%%%%%%%%%%%%%%%%%%%%%%%%%%%%%%%  Proof Of theorem 1
\begin{proof}[Proof of Theorem \ref{th:theorem1}]
  Since the two users have the opportunity to overhear a noisy
  version of each other's signal, they set up a cooperation protocol
  in transmitting their messages to the receiver. To furnish the
  cooperation, we let each user \emph{partially} decodes the other
  user's transmitted message, facilitating coherent transmission of
  part of their messages. The other parts of the messages are
  transmitted directly to the receiver. Hence, the $k$th encoder
  splits its message into two parts as $m_k=(m_{k,3-k},m_{k3})$, drawn
  uniformly from the sets $[1:2^{nR_{k,3-k}}]\times [1:2^{nR_{k3}}]$,
  where $R_k=R_{k,3-k}+R_{k3}$, for $k\in\{1,2\}$. Thus, $m_{k,3-k}$
  is transmitted to the receiver with the help of user $(3-k)$ by
  use of BME combined with the GPC technique and $m_{k3}$ denotes the
  message that is transmitted directly to the receiver by use of
  superposition coding and the GPC technique.

  Each user transmits $B$ blocks of length $n$ with large enough
  block length to ensure reliable decoding at the other user and the
  receiver. User $k$ transmits a sequence of $B-1$ messages over
  $B$ blocks. Therefore, for a fixed $n$, the effective transmitted
  rate is $\big((R_{k,3-k}+R_{k3})(B-1)/B\big)$, that approaches to $R_k$ as
  $B$ tends to infinity. Moreover, the receiver uses backward
  decoding technique \cite{willemsThesis} to decode the transmitted
  messages. Although one drawback of backward decoding is that the receiver needs to wait $B$
  blocks, the rate region is larger than that corresponding to low-delay sliding window decoding
  \cite{carleial1982multiple}. In the following for $k\in\{1,2\}$ and for each block $b \in [1:B]$,
  we explain our proposed three-layer code construction at the $k$th
  encoder.
%%%%%%%%%%%%%%%%%%%%%%%%%%%%%%%%%%%%%%%%%%%%%%           Figure: Codebooks
%\begin{flushright}
\begin{figure*}[!t]
\centering
\psfrag{d12}[][][0.6]{$\ddots$}
\psfrag{d}[][][0.6]{$\vdots$}
\psfrag{d2}[][][0.6]{$\cdots$}
\psfrag{1}[][][0.85]{$1$}
\psfrag{2}[][][0.85]{$2$}
\psfrag{3}[][][0.85]{$\doublehat{m}^b_c$}
\psfrag{5}[][][0.85]{$m^b_{12}$}
\psfrag{6}[][][0.85]{$2^{nR_{12}}$}
\psfrag{7}[][][0.85]{$m^b_{13}$}
\psfrag{8}[][][0.85]{$2^{nR_{13}}$}
\psfrag{4}[][][0.85]{$2^{n(R_{12}+R_{21})}$}
\psfrag{s0}[][][0.85]{$\bs_0$}
\psfrag{s1}[][][0.85]{$(\bs_0,\bs_1)$}
\psfrag{u}[][][0.85]{$\bu(\doublehat{m}^b_c,j^*_0)$}
\psfrag{v1}[][][0.85]{$\bv_1(\doublehat{m}^b_c,j^*_0,m^b_{12},j^*_{1})$}
\psfrag{v3}[][][0.85]{$\bv_{13}(\doublehat{m}^b_c,j^*_0,m^b_{12},j^*_{1},m^b_{13},j^*_{13})$}
\includegraphics[scale = 0.35]{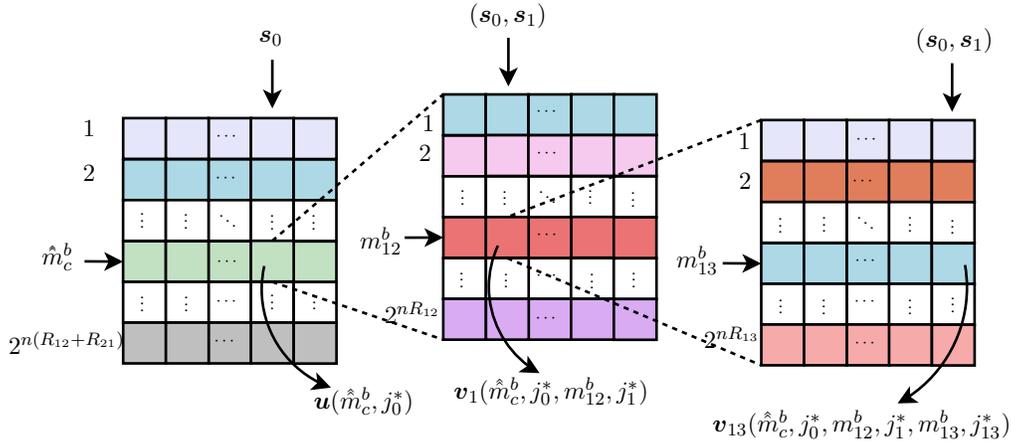}
\caption{\small Illustration of the three-layer codebook generation at encoder 1. It is shown that for each layer how codebooks are generated and how codewords are chosen based on the respective message and given CSIT by the GPC technique. For each layer, a row indicates a bin which is determined by the message, and then for a selected bin, a codeword is chosen which is jointly typical with the given CSIT.}
\label{fig:codebook}
\vspace{-.7cm}
\end{figure*}
%\end{flushright}
%%%%%%%%%%%%%%%%%%%%%%%%%%%%%%%%%%%%%%%%%%%%%%%%%%%%%%%%%%%%%%%%%%%%%%%%%%%%%%%%%%%%%%%%%%%%%%%%%%%%%%%%%%%%%%%%%%%%%%%%%%%%%%%

  \textbf{Codebook generation}: For a fixed pmf of the form given in
  (\ref{eq: pmf}) and a given $\epsilon>0$, for $k\in\{1,2\}$ let
\begin{align*}
  J_0&=2^{n[I(U;S_0)+\epsilon]},~~ J_{k}=2^{n[I(V_k;S_0S_k|U)+\epsilon]},~~J_{k3}=2^{n[I(V_{k3};S_0S_k|UV_k)+\epsilon]}.
\end{align*}
\begin{itemize}
  \item \emph{Transmit old information coherently}:  Generate $2^{n(R_{12}+R_{21})}\times J_0$ codewords $\bu$, each with probability $p(\bu)=\prod_{t=1}^{n}{p_{U}(u_{t})}$. Label them as $\bu(m_c,j_0)$ for $m_c\in[1:2^{n(R_{12}+R_{21})}]$ and $j_0\in[1:J_0]$.
  \item \emph{Superimpose fresh information}: For each pair $\bu(m_c,j_0)$, $k$th user generates $2^{nR_{k,3-k}}\times J_k$ codewords $\bv_k$, each with probability $p(\bv_k|\bu)=\prod_{t=1}^{n}{p_{V_k|U_0}(v_{k,t}|u_{0,t})}$. Label them as $\bv_k(m_c,j_0,m_{k,3-k},j_k)$ where $m_{k,3-k}\in [1:2^{nR_{k,(3-k)}}]$ and $j_k\in [1:J_k]$.
  \item \emph{Direct transmission to the receiver}: User $k$, for each pair of codewords $\big(\bu(m_c,j_0),\bv_k(m_c,j_0,m_{k,3-k},j_k)\big)$ generates $2^{nR_{k,3}}\times J_{k3}$ codewords $\bv_{k3}$ each with probability $p(\bv_{k3}|\bu,\bv_k)=$ $\prod_{t=1}^{n}{p_{V_{k3}|UV_k}(v_{k3,t}|u_{0,t},v_{k,t})}$. Label them as $\bv_{k3}(m_c,j_0,m_{k,3-k},j_k,m_{k,3},j_{k3})$ where $m_{k,3}\in [1:2^{nR_{k3}}]$ and $j_{k3}\in[1:J_{k3}]$.
\end{itemize}
Encoders 1 and 2, by use of the generated codebooks, the noncausal CSIT and the messages, choose appropriate codewords and generate transmitted signals. In the following, we explain how the encoders generate codewords for block $b$ and then the receiver performs decoding procedure.

\textbf{Generating $\bx_1^b$ at encoder 1:}
Assume that encoder 1, from decoding of the previous block, has already estimated correctly the first part of the transmitted message of encoder 2 as $\doublehat{m}^{b-1}_{21}$. Therefore, encoder 1 transmits the cooperation message $\doublehat{m}_c^b:=(m^{b-1}_{12},\doublehat{m}_{21}^{b-1})$ over $b$th block.
\begin{itemize}
  \item \emph{Common codeword}: By use of $\doublehat{m}_c^b$ and known common CSIT $\bs^b_0$, encoder 1 searches for smallest $j_0$ such that $\big(\bu(\doublehat{m}_c^b,j_0),\bs^b_0\big)\in \A (p_{US_0})$. Name this $j_0$ as $j^*_0$ which is a function of $(\doublehat{m}_c^b,\bs^b_0)$.

  \item \emph{Superimpose fresh information}: To transmit $m^b_{12}$, for a given $\bu(\doublehat{m}_c^b,j^*_0)$ and known CSIT $(\bs^b_0,\bs^b_1)$, encoder 1 searches for smallest $j_1$ such that $\bv_1(\doublehat{m}_c^b,j^*_0,m^b_{12},j_1)$ is jointly typical with $(\bs^b_0,\bs^b_1,\bu(\doublehat{m}_c^b,j^*_0))$. Name this $j_1$ as $j^*_1$ which is a function of $(\doublehat{m}_c^b,m^b_{12},\bs^b_0,\bs^b_1)$.

  \item \emph{Direct transmission}: To transmit $m^b_{13}$ directly to the receiver, encoder 1 searches for smallest $j_{13}$ such that $\bv_{13}(\doublehat{m}_c^b,j^*_0,m^b_{12},j^*_1,m^b_{13},j_{13})$ is jointly typical with $(\bs^b_0,\bs^b_1,\big(\bu(\doublehat{m}_c^b,j^*_0),\bv_1(\doublehat{m}_c^b,j^*_0,m^b_{12},j^*_1)\big)$. Name this $j_{13}$ as $j^*_{13}$ which is a function of $(\doublehat{m}_c^b,m^b_{12},m^b_{13},\bs^b_0,\bs^b_1)$.

  \item \emph{Transmitted signal}: Finally, for given codewords $\big(\bu(\doublehat{m}_c^b,j^*_0),\bv_1(\doublehat{m}_c^b,j^*_0,m^b_{12},j^*_1),\bv_{13}(\doublehat{m}_c^b,j^*_0,m^b_{12},$ $j^*_1,m^b_{13},j^*_{13})\big)$ and known CSIT $(\bs^b_0,\bs^b_1)$, encoder 1 generates $\bx^b_1$ with independent and identically distributed (i.i.d.) component according to the conditional pmf $p_{X_1|S_0S_1UV_1V_{13}}$. The codebook generation flow of encoder 1 is depicted in Fig.~\ref{fig:codebook}.
\end{itemize}
\hspace{0.13cm}\textbf{Generating $\bx_2^b$ at encoder 2:}
$\bx_2^b$ is generated similar to $\bx_1^b$ by swapping the indices 1 and 2.

Note that in the codebook generation we set $m^0_{12}=m^0_{21}=m^1_{13}=m^1_{23}=m^B_{12}=m^B_{21}=1$.\\
%%%%%%%%%%%%%%%%%%%%%%%%%%%%%%%%%%%%%%%%%%%%%%%%%%%%%%%%%%%%%%%%%%%%%%%%%%%%%%%%%%%%%%%%%%%%%%%%%%%%%%%%%%%%%%%%%%%%%%
%%%%%%%%%%%%%%%%%%%%%%%%%%%%%%%%%   Table of coding and backward decoding
\begin{table*}\footnotesize
\caption{\small Illustration of encoding--decoding flow at encoder 1 and the backward decoding at the receiver. The superscript indices indicate the block number.  }\label{tabel:encoding}
\begin{center}
\begin{tabular}{ >{\columncolor[gray]{.9}}l | c| c | c| c}\hline
\rowcolor [gray]{.7} Block & 1 & 2& $\cdots$ & $B$  \\ \hline
$U_0$ & $\bu^1\big((1,1),j^1_0\big)$ & $\bu^2\big((m^1_{12},\doublehat{m}^1_{21}),j^2_0\big)$& $\cdots$ & $\bu^B\big((m^{B-1}_{12},\doublehat{m}^{B-1}_{21}),j^{B}_0\big)$  \\
$V_1$ & $\bv^1_1\big((1,1),j^1_0,m^1_{12},j^1_1\big)$ & $\bv^2_1\big((m^1_{12},\doublehat{m}^1_{21}),j^2_0,m^2_{12},j^2_1\big)$& $\cdots$ & $\bv^B_1\big((m^{B-1}_{12},\doublehat{m}^{B-1}_{21}),j^B_0,1,j^B_1\big)$ \\
$V_{13}$ & $\bv_{13}\big((1,1),j^1_0,m^1_{12},j^1_1,1,j^1_{13}\big)$ & $\bv^2_{13}\big((m^1_{12},\doublehat{m}^1_{21}),j^2_0,m^2_{12},j^2_1,m^2_{13},j^2_{13}\big)$ & $\cdots$ & $\bv^B_{13}\big((m^{B-1}_{12},\doublehat{m}^{B-1}_{21}),j^B_0,1,j^B_1,m^B_{13},j^B_{13}\big)$ \\
$Y_1$ & $\doublehat{m}^1_{21}\longrightarrow$ & $\doublehat{m}^2_{21}\longrightarrow$ & $\cdots$ & $\emptyset$  \\
$Y_3$ & $\emptyset$ & $\longleftarrow\big(\hat{m}^1_{12},\hat{m}^1_{21},\hat{m}^2_{13},\hat{m}^2_{23}\big)$ & $\cdots$ & $\longleftarrow\big(\hat{m}^{B-1}_{12},\hat{m}^{B-1}_{21},\hat{m}^B_{13},\hat{m}^B_{23}\big)$  \\
\bottomrule
\end{tabular}
\end{center}
\vspace{-.5cm}
\end{table*}
%%%%%%%%%%%%%%%%%%%%%%%%%%%%%%%%%%%%%%%%%%%%%%%%%%%%%%%%%%%%%%%%%%%%%%%%%%%%%%%%%%%%%%%%%%%%%%%%%%%%%%%%%%%%%%%%%%%%%%
\textbf{{Decoding}:} In addition to the decoding process at the
receiver, each encoder also decodes partially the transmitted message
of the other encoder. All the decoding is based on joint
typicality of sequences which are discussed in the following.
\begin{itemize}
\item \emph{Decoding at encoder 1}: Assume that encoder 1 has received
  $\by^b_1$ for $b\in[1:B-1]$ and $\doublehat{m}^{b-1}_{21}$ is
  already known from the decoding process of the previous block. In
  other words,
  $\doublehat{m}^b_c:=(m^{b-1}_{12},\doublehat{m}^{b-1}_{21})$ is
  available. Given
  $(\bu(\doublehat{m}_c^{b},j^*_0),\bs^b_0,\bs^b_1,\bx^b_1)$, which are
  known codewords at encoder 1, encoder 1 tries to estimate the first part
  of the transmitted message of encoder 2 over $b$th block as
  $\doublehat{m}^b_{21}$. So, encoder 1 checks the following joint
  typicality of sequences to decode the message.
\begin{align}
\label{eq: decodingEn1}
\left(\bu(\doublehat{m}_c^{b},j^*_0),\bv_2(\doublehat{m}_c^{b},j^*_0,m_{21}^{b},j_2),\bs^b_0,\by^b_1,\bs^b_1,\bx^b_1\right)\in \A(UV_2S_0Y_1X_1S_1).
\end{align}
  \item \emph{Decoding at encoder 2}: Decoding is similar to the decoding at encoder 1 by swapping indices.
  \item \emph{Decoding at the receiver}: The receiver uses backward
    decoding \cite{willemsThesis} to decode the transmitted messages,
    i.e., the receiver waits to receive all $B$ blocks, and then
    starts decoding the transmitted messages from the last to the
    first block. In the following, we explain the decoding at block
    $b$. From the decoding of block $b+1$, the receiver already knows
    $\hat{M}^b_{12}:=\hat{m}^b_{12}$ and
    $\hat{M}^b_{21}:=\hat{m}^b_{21}$. Therefore, the receiver
    estimates $m^b_c=(m^{b-1}_{12},m^{b-1}_{21})$, $m^b_{13}$ and
    $m^b_{23}$ using the following joint typicality criterion for
    $b=B,B-1,...,2$;
\begin{equation}
\label{eq: RXdecoding}
E^b_D \in \A(D),
\end{equation}
where
\begin{eqnarray*}
E^b_D &:=& \big(\bu(m_c^b,j_0),\bv_1(m_c^b,j_0,\hat{M}_{12}^{b},j_1),\bv_2(m_c^b,j_0,\hat{M}_{21}^{b},j_2),\bv_{13}(m_c^b,j_0,\hat{M}_{12}^{b},j_1,m^b_{13},j_{13}),\\
&& ~~\bv_{23}(m_c^b,j_0,\hat{M}_{21}^{b},j_2,m^b_{23},j_{23}),\by^b_3 \big),\\
D &:=& (U,V_1,V_2,V_{13},V_{23},Y_3).
\end{eqnarray*}
\end{itemize}
The encoding-decoding flow of encoder 1 and the backward decoding flow at the receiver are summarized in Table~\ref{tabel:encoding}. Analysis of probability of decoding errors are given in Appendix I.
\end{proof}
%%%%%%%%%%%%%%%%%%%%%%%%%%%%%%%%%%%%%%%%%%%%%%%%%%%%%%%%%%%%%%%%%%%%%%%%%%%%%%%%%%%%%%%%%%%%%%%%%%%%%%%%%%%%%%%%%%%%%%%%%%%%%%%%%%%%
%%%%%%%%%%%%%%%%%%%%%%%%%%%%%%%%%%%%%%%%%%%%%%%%%%%%%%%%%%%%%%%%%%%%%%%%%%%%%%%%%%%%%%%%%%%%%%%%%%%%%%%%%%%%%%%%%%%%%%%%%%%%%%%%%%%%%%%%%%%%%%%%%
%%%%%%%%%%%%%%%%%%%%%%%%%%%%%%%%%%%%%%%%%%%%%%%%%%%   Gaussian Figure
\label{channelmodel}
\begin{flushleft}
\begin{figure}[!t]
\centering
\psfrag{m1}[][][.85]{$m_1$}
\psfrag{m2}[][][.85]{$m_2$}
\psfrag{m12}[][][.85]{$(\hat{m}_1,\hat{m}_2)$}
\psfrag{s0}[][][.85]{$S_0$}
\psfrag{s1}[][][.85]{$S_1$}
\psfrag{s2}[][][.85]{$S_2$}
\psfrag{x1}[][][.85]{$X_1$}
\psfrag{x2}[][][.85]{$X_2$}
\psfrag{z3}[][][.85]{$Z_3$}
\psfrag{z1}[][][.85]{$Z_1$}
\psfrag{z2}[][][.85]{$Z_2$}
\psfrag{y1}[][][.85]{$Y_1$}
\psfrag{y2}[][][.85]{$Y_2$}
\psfrag{y3}[][][.85]{$Y_3$}
\psfrag{encoder 1}[][][.85]{Encoder 1}
\psfrag{encoder 2}[][][0.85]{Encoder 2}
\psfrag{decoder}[][][0.85]{Decoder}
\includegraphics[scale=0.25]{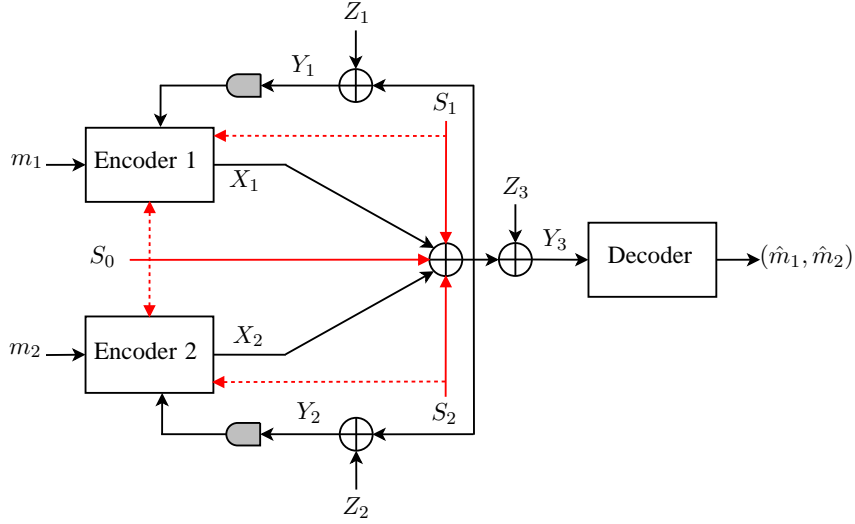}
\caption{\small The two-user Gaussian GMAC with additive interference.}
\label{fig:Gaussiancnlmdl}
\vspace{-0.7cm}
\end{figure}
\end{flushleft}
%%%%%%%%%%%%%%%%%%%%%%%%%%%%%%%%%%%%%%%%%%%%%%%%%%%%%%%%%%%%%%%%%%%%%%%%%%%%%%%%%%%%%%%%%%%%%%%%%%%%%%%%%%%%%%%%%%%%%%%%%%%%%%%%%%%%%%%%%%%%%%%%%
%\vspace{-1.7cm}
\section{State-Dependent Gaussian GMAC}
In this section, a dirty GMAC is studied for two scenarios. First, it is assumed that both of encoders know the full CSI, and then it is assumed that each of the them only knows part of the CSI. Assume a dirty GMAC as depicted in Fig.~\ref{fig:Gaussiancnlmdl} which is defined by a linear Gaussian channel model as
\begin{equation}
\label{eq:gaussianmodel}
Y_j=X_1+X_2+S_0+S_1+S_2+Z_j,
\end{equation}where $Y_j$ is the received signal at the node $j$ for $j\in\{1,2,3\}$, which is corrupted by the AWGN noise $Z_j\sim\mathcal{N}(0,N_j)$ and the additive interferences. The common interference signal $S_0$ is noncausally known at both of the encoders, and $S_k$ is only available at the $k$th encoder, for $k\in{1,2}$. We assume that the additive interference $S_0$ has a zero mean Gaussian distribution with a \emph{finite} variance $\mathbb{E}S_0^2 = Q_0$, i.e. $S_0 \sim \mathcal{N}(0,Q_0)$, and the private additive interference is $S_k\sim\mathcal{N}(0,Q_k)$. The additive AWGN noises and interferences are mutually independent. The transmitted signal from the $k$th encoder is denoted by $X_k$ subject to an average power constraint $\mathbb{E}X_k^2\leq P_k$ for $k\in\{1,2\}$.
\subsection{State-Dependent Gaussian GMAC with fully informed encoders}
In this subsection, we consider a dirty Gaussian GMAC in which both of encoders knows the full CSI. That is, $S_1=S_2=\emptyset$ and the linear Gaussian channel model is defined by
\begin{equation}
\label{eq:gaussianmodel2}
Y_j=X_1+X_2+S_0+Z_j.
\end{equation}
\begin{proposition}\label{propos:proposition1} \emph{For the Gaussian GMAC with full CSI at the encoders defined in (\ref{eq:gaussianmodel2}), the rate region
\begin{align}\label{eq: proposition1}\nonumber
\mathcal{R}_G:= \bigcup
\Bigg\{(R_1,R_2): R_1&\leq \c \left(\frac{\bar{\rho}_1P'_1}{\bar{\rho}_1P''_1+N_2}\right)+\c \left(\frac{\bar{\rho}_1P''_1}{N_3}\right),\\\nonumber
R_2&\leq  \c \left(\frac{\bar{\rho}_2P'_2}{\bar{\rho}_2P''_2+N_1}\right)+\c \left(\frac{\bar{\rho}_2P''_2}{N_3}\right),\\\nonumber
R_1+R_2&\leq \c \left(\frac{\bar{\rho}_1P'_1}{\bar{\rho}_1P''_1+N_2}\right)+\c \left(\frac{\bar{\rho}_2P'_2}{\bar{\rho}_2P''_2+N_1}\right)+
\c \left(\frac{\bar{\rho}_1P''_1+\bar{\rho}_2P''_2}{N_3}\right)\\
R_1+R_2&\leq \c \left(\frac{P_1+P_2+2\sqrt{\rho_1\rho_2P_1P_2}}{N_3}\right)\Bigg\},
\end{align}
is achievable, where the union is taken over $\rho_k \in[0,1]$ and $P'_k,P''_k \geq 0$ such that $P'_k+P''_k\leq P_k$ for $k\in\{1,2\}$ and
$\c(x):=\frac{1}{2}\log_2(1+x)$, $\bar{x}:=1-x$.}
\end{proposition}
\begin{remark}\emph{The achievable rate region of the state-dependent Gaussian GMAC with full CSIT is independent of the additive interference.
That is, the effect of the additive interference is completely removed as if there is no interference over the channel (clean channel). Hence, the region is equal to that of the counterpart channel without interference \cite[Chapter 11]{kramer2007topics}. Moreover, the achievable rate region (\ref{eq: proposition1}) includes that of the partial decode-and-forward Gaussian relay channel with no interference \cite{ZahediIT2006}}.
\end{remark}
\begin{proof}[Proof of Proposition \ref{propos:proposition1}] To prove the result we use the rate region developed in Theorem \ref{th:theorem1}. Note that we can extend the result in Theorem \ref{th:theorem1} for the discrete case to the Gaussian model in a similar approach given in \cite[Chapter 23]{elgamal2011network}. For $k\in \{1,2\}$, let
\begin{subequations}\label{eq:txsignals}
\begin{align}
U &= X_u+\alpha_0 S_0,\\
V_k&= \sqrt{\rho_kP_k‎}‎X_u+X'_k+\alpha_kS_0,\\
V_{k3}&=X''_k+\alpha_{k3}S_0,
%V_1&= \sqrt{\rho_1P_1‎}‎X_u+X'_1+\alpha_1S_0,\\
%V_2&= \sqrt{\rho_2P_2}X_u+X'_2+\alpha_2S_0,\\
%V_{13}&=X''_1+\alpha_{13}S_0,\\
%V_{23}&=X''_2+\alpha_{23}S_0,
\end{align}
\end{subequations}
where $X_u\sim\mathcal{N}(0,1)$ carries the cooperation message to the receiver, the fresh information and the
private message are transmitted via $X'_k\sim\mathcal{N}(0,\bar{\rho}_k P'_k)$ and $X''_k\sim\mathcal{N}(0,\bar{\rho}_k P''_k)$
by the $k$th encoder, respectively. The random variables $X_u,X'_k$ and $X''_k$ are mutually independent
for $k\in\{1,2\}$ and are independent from $S_0$. Finally, let the transmitted symbols at the encoders be
\begin{align}\label{eq: x_k}
X_k&=\sqrt{\rho_k P_k}X_u+X'_k+X''_k.
\end{align}
To ensure the average power constraint at each encoder, we have $P'_k+P''_k\leq P_k$. In (\ref{eq:txsignals}), $\alpha_0,\alpha_k$ and $\alpha_{k3}$ for $k\in\{1,2\}$, are design parameters to be optimized. In Appendix II, we prove that the optimal choices of the precoding (i.e., dirty paper coding (DPC)) coefficients in (\ref{eq:txsignals}) are given by
\begin{subequations}\label{eq: DPCcoeff}
\begin{align}
\alpha_0 &= \frac{\sqrt{\rho_1P_1}+\sqrt{\rho_2P_2}}{P_1+P_2+2\sqrt{\rho_1\rho_2P_1P_2}+N_3},\\
\alpha_k &= \frac{\sqrt{\rho_k P_k}\big(\sqrt{\rho_1P_1}+\sqrt{\rho_2P_2}\big)+\bar{\rho}_k P'_k}{P_1+P_2+2\sqrt{\rho_1\rho_2P_1P_2}+N_3},\\
\alpha_{k3} &= \frac{\bar{\rho}_k P''_k}{P_1+P_2+2\sqrt{\rho_1\rho_2P_1P_2}+N_3}.
\end{align}
\end{subequations}
Now by substituting the auxiliary random variables given in (\ref{eq:txsignals}), (\ref{eq: x_k}), the above optimal precoding coefficients in the achievable rate region of Theorem \ref{th:theorem1}, noting that $\Delta^-=\delta^-_1=\delta^-_2=0$, and utilizing Fourier--Motzkin elimination algorithm \cite{elgamal2011network}, the achievable rate region in Proposition \ref{propos:proposition1} is derived.
\end{proof}
%\subsubsection{Numerical evaluation: Bounds on the capacity of symmetric Gaussian GMAC with fully informed encoders}
To investigate the role of cooperation between the encoders, a Gaussian dirty GMAC with full CSIT is considered wherein $P_1=P_2=10$ [dB], $N_1=N_1=0$ [dB] and $N_3=7$ [dB] for $Q_0\in\{2,5,8\}$ [dB]. The achievable rate regions are depicted in Fig.~\ref{fig:numerical} for the following four scenarios:
\begin{itemize}
  \item \emph{Dirty GMAC with full CSIT}: The encoders cooperate with each other and the effect of additive interference is completely removed;
  \item \emph{Dirty MAC with full CSIT}: A conventional Gaussian MAC and the effect of additive interference is completely removed. By substituting $\rho_k=P'_k=0$ for $k\in\{1,2\}$ in (\ref{eq: proposition1}), the associated rate region is established.
  \item \emph{Dirty GMAC without CSIT}: A Gaussian GMAC where the interference is treated as an additive noise; and
  \item \emph{Dirty MAC without CSIT}: A conventional Gaussian MAC where the interference is treated as an additive noise.
\end{itemize}
From Fig. ~\ref{fig:numerical}, we observe that the achievable rate region of the channel model with cooperating encoders and CSIT is the largest one and that of the case without cooperating encoders (i.e., conventional MAC) in which the CSI is \emph{not} available at the encoders is the worst case. As it is shown in Figs.~\ref{fig:numerical}(a)--\ref{fig:numerical}(c), when the power of the interference decreases, the achievable rate region for the case with cooperating encoders \emph{without} CSIT becomes larger than that of the conventional MAC \emph{with} CSIT.
%%%%%%%%%%%%%%%%%%%%%%%%%%%%%%%%%%%%%%%%%%%%%%%%%%%   Gaussian Numerical results
\begin{figure}[!t]
        \centering
        \psfrag{R1}[][][.6]{$R_1$}
        \psfrag{R2}[][][.6]{$R_2$}
        \psfrag{A}[][][.6]{$\times$A}
        \psfrag{B}[][][.6]{$\times$B}
        \begin{subfigure}[b]{0.32\textwidth}
                \centering
                \includegraphics[scale=0.4]{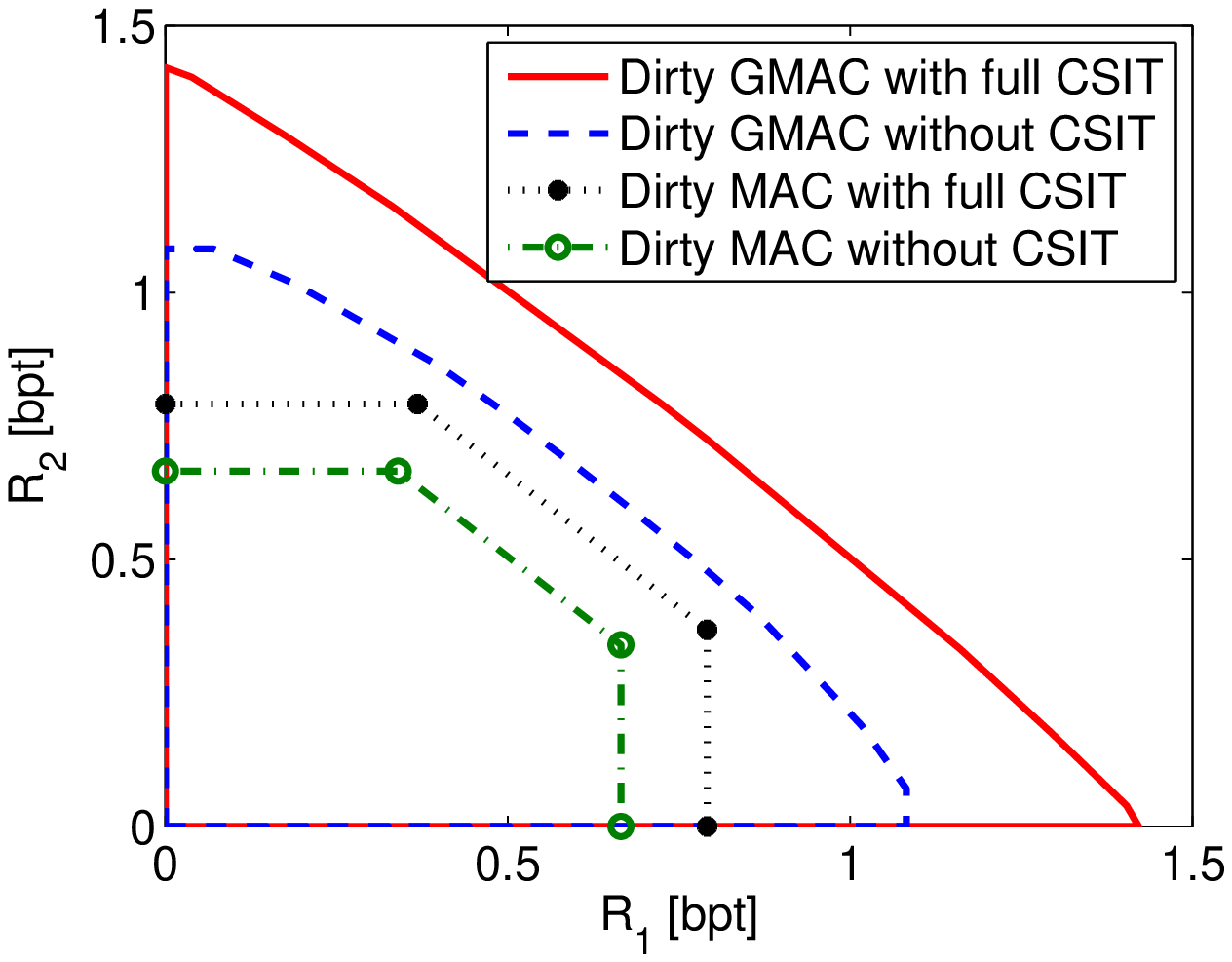}
                \caption{$Q=2$ [dB]}
                \label{fig:q2}
        \end{subfigure}%
        ~ %add desired spacing between images, e. g. ~, \quad, \qquad etc.
        \begin{subfigure}[b]{0.3\textwidth}
                \centering
                \includegraphics[scale=0.4]{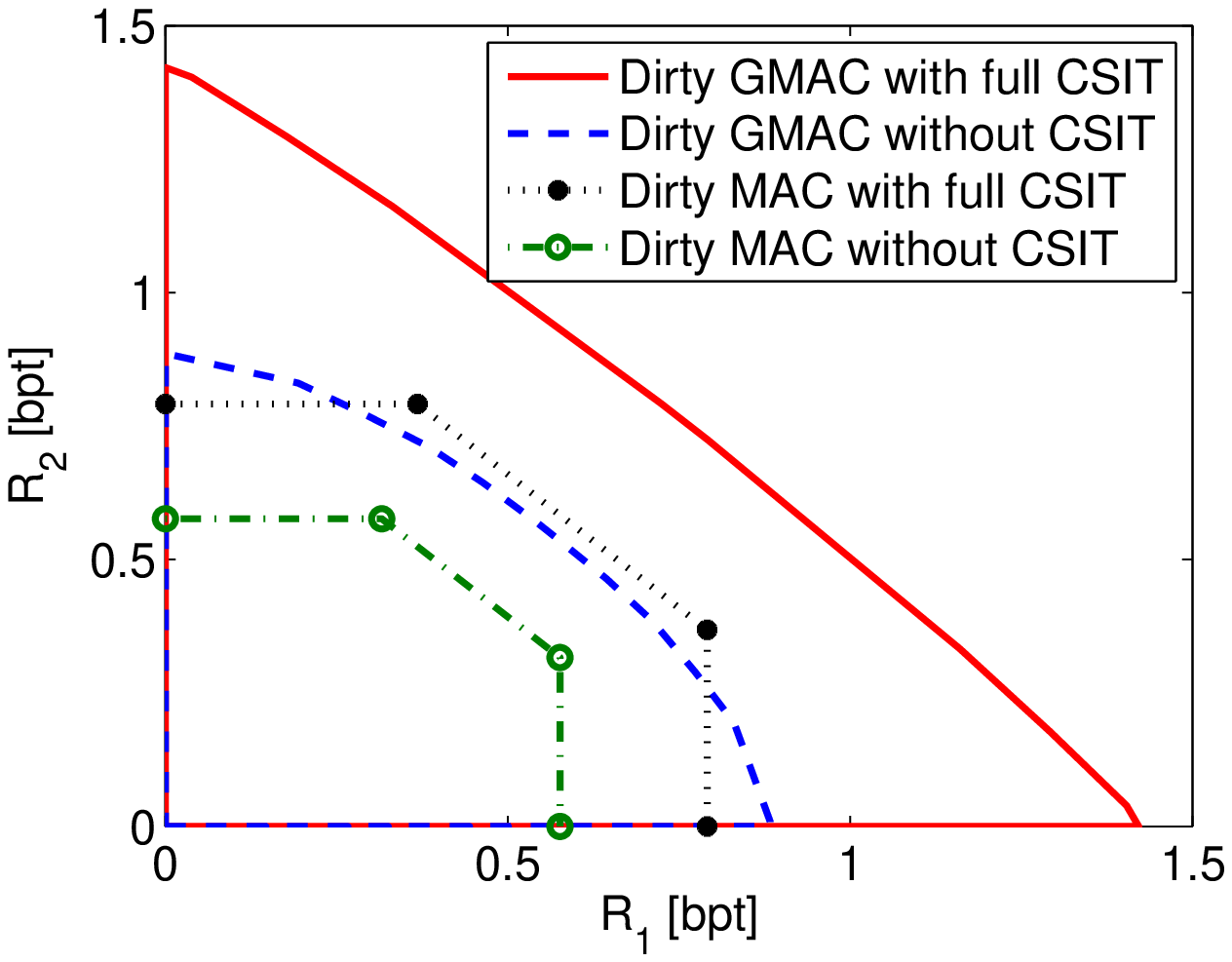}
                \caption{$Q=5$ [dB]}
                \label{fig:q5}
        \end{subfigure}
        ~ %add desired spacing between images, e. g. ~, \quad, \qquad etc.
        \begin{subfigure}[b]{0.3\textwidth}
                \centering
                \includegraphics[scale=0.4]{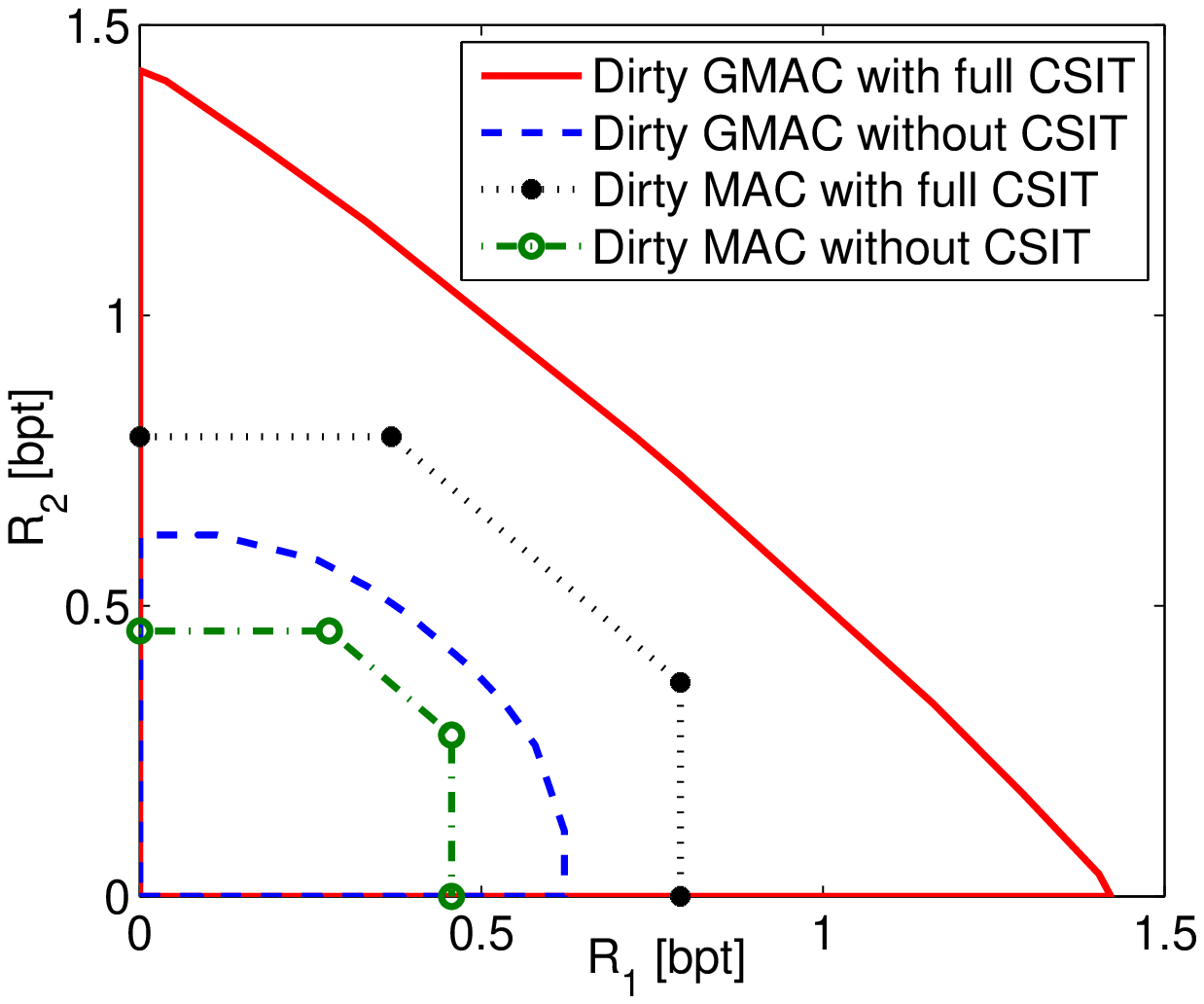}
                \caption{$Q=8$ [dB]}
                \label{fig:q8}
        \end{subfigure}
        \caption{\small Achievable rate regions for the state-dependent GMAC with full CSI at the encoders for given channel parameters $P_1=P_2=10\text{ [dB]}, N_1=N_2=0\text{ [dB]}$, $N_3=7\text{ [dB]}$ and $Q \in \{2,5,8\}$ [dB].}\label{fig:numerical}
        \vspace{-0.5cm}
\end{figure}
%%%%%%%%%%%%%%%%%%%%%%%%%%%%%%%%%%%%%%%%%%%%%%%%%%%%%%%%%%%%%%%%%%%%%%%%%%%%%%%%%%%%%%%%%%%%%%%%%%%%%%%%%%%%%%%%%%%%%%%%%%%%%%%%%%%%%%%%%%%%%%%%%
\subsection{State-Dependent Gaussian GMAC with partially informed encoders}
In this subsection, we evaluate the achievable rate region of the general dirty Gaussian GMAC defined in (\ref{eq:gaussianmodel}) and one special case. It is assumed that $(S_0,S_k)$ are available noncausally at the $k$th encoder. With a similar discussion as that in the previous subsection, by an optimal choice of Costa precoding coefficients the effect of $S_0$ can be completely removed. Without loss of generality, a doubly dirty GMAC defined in (\ref{eq:gaussianmodel}) subject to $S_0=\emptyset$ is studied in the following. Before presenting the achievable rate region, the procedure of generating signals at encoder $k$ is discussed in the following. An illustration of the signal generation at encoder 1 is depicted in Fig.~\ref{fig:GaussianDPC1}.

The transmitted signal from each encoder has four components:
\begin{itemize}
\item \emph{Partial cleaning of the private state:} Since $S_k$ is only available at encoder $k$, the encoder dedicates $\bar{\eta_k}P_k$ of its total power to clean the state, where $1-\min(1,Q_k/P_k)\leq \eta_k \leq 1$. Therefore, the effective additive interference over the channel becomes $\beta_kS_k$ where $\beta_k:=1-\sqrt{\tfrac{\bar{\eta_k}P_k}{Q_k}}$.
\item \emph{Message cooperation}: Gaussian codeword $U\sim\mathcal{N}(0,1)$ carries the cooperation message $m_c^b$.
\item \emph{Fresh information}: Gaussian codeword $X'_k\sim\mathcal{N}(0,P'_{ke})$ which is generated by use of generalized DPC (GDPC), carries the fresh message $m_{k,3-k}^b$, where $P'_{ke}=\eta_k \bar{\rho_k}P'_k$. By use of GDPC, we mean that, the state is partially cleaned and then DPC is used to combat $\beta_kS_k$.
\item \emph{Direct transmission}: Gaussian codeword $X''_k\sim\mathcal{N}(0,P''_{ke})$ which is generated by use of GDPC, carries the private message $m_{k3}^b$, where $P''_{ke}=\eta_k \bar{\rho_k}P''_k$.
\end{itemize}
\begin{figure}[!t]
\centering
\psfrag{randomcoding}[][][.83]{Random coding}
\psfrag{CostaCoding}[][][.85]{Costa coding}
\psfrag{mc}[][][.8]{$m_c^b$}
\psfrag{m12}[][][.8]{$m_{12}^b$}
\psfrag{m13}[][][.8]{$m_{13}^b$}
\psfrag{s1}[][][.8]{$S_1^b$}
\psfrag{factor}[][][.8]{$-\sqrt{\tfrac{\bar{\eta_1}P_1}{Q_1}}$}
\psfrag{f2}[][][.8]{$\sqrt{\eta_1\rho_1P_1}$}
\psfrag{u}[][][.8]{$U(m_c^b)$}
\psfrag{xp}[][][.8]{$X'_1(m_{12}^b,S_1^b)$}
\psfrag{xz}[][][.8]{$X''_1(m_{13}^b,S_1^b)$}
\psfrag{x1}[][][.8]{$X_1(m_c^b,m_{12}^b,m_{13}^b,S_1^b)$}
\includegraphics[scale=0.32]{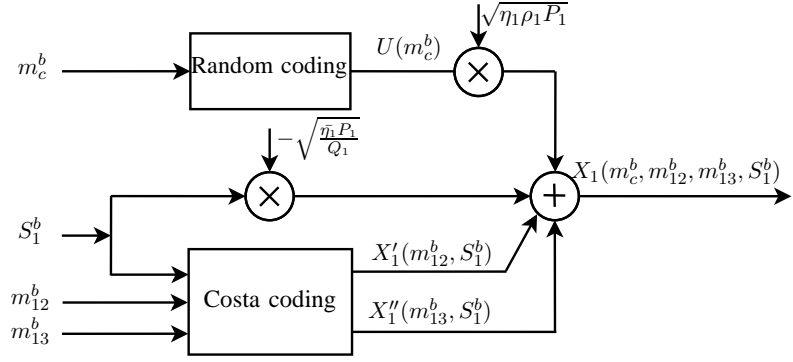}
\caption{\small Signal generation for the doubly dirty GMAC at encoder 1.}
\label{fig:GaussianDPC1}
\vspace{-0.5cm}
\end{figure}

The auxiliary random variables and the transmitted signal are more specifically given by
\begin{subequations}\label{eq: AuxRvDGMAC}
\begin{align}
V_k&:=\sqrt{\eta_k\rho_kP_k}U+X'_k+\alpha_k \beta_k S_k,\\
V_{k3}&:=X''_k+\alpha_{k3} \beta_k S_k,\\
X_k&:=\sqrt{\eta_k\rho_kP_k}U+X'_k+X''_k+(\beta_k-1)S_k.
\end{align}
\end{subequations}
To ensure the average power constraint at each encoder, we have $P'_k+P''_k\leq P_k$. All auxiliary random variables $U,~X'_k,~X''_k$, and $S_k$ are mutually independent from each other and the channel noises.
\begin{proposition}\label{propos:proposition2} \emph{For the general dirty Gaussian GMAC defined in (\ref{eq:gaussianmodel}) with partially informed encoders, the rate region $\mathcal{R}^d_G$ is achievable;
\small
\begin{align}\nonumber
\mathcal{R}^d_G:= \bigcup \Bigg\{(R_1,R_2):& R_1=R_{12}+R_{13},~R_2=R_{21}+R_{23},\\\nonumber
R_{12} \leq & \c \left(\frac{P'_{1e}+Q_{1e}-\hat{Q}_{1e}}{P''_{1e}+\hat{Q}_{1e}+N_2}\right)-\c_{12},\\\nonumber
R_{21}\leq & \c \left(\frac{P'_{2e}+Q_{2e}-\hat{Q}_{2e}}{P''_{2e}+\hat{Q}_{2e}+N_1}\right)-\c_{21},\\\nonumber
R_{13}\leq&\c \left(\frac{P''_{1e}+\hat{Q}_{1e}-\doublehat{Q}_{1e}}{\doublehat{Q}_{1e}+\doublehat{Q}_{2e}+N_3}\right)-
\c_{13}+\delta_1^-,\\\nonumber
R_{23}\leq&\c \left(\frac{P''_{2e}+\hat{Q}_{2e}-\doublehat{Q}_{2e}}{\doublehat{Q}_{1e}+\doublehat{Q}_{2e}+N_3}\right)-
\c_{23}+\delta_2^-,\\\nonumber
R_{13}+R_{23}\leq&\c \left(\frac{P''_{1e}+P''_{2e}+\hat{Q}_{1e}-\doublehat{Q}_{1e}+\hat{Q}_{2e}-\doublehat{Q}_{2e}}{\doublehat{Q}_{1e}+\doublehat{Q}_{2e}+N_3}\right)-
\c_{13}-\c_{23}+\Delta^-,\\
R_1+R_2\leq&\c \left(\frac{\eta_1P_1+\eta_2P_2+2\sqrt{\eta_1\eta_2\rho_1\rho_2P_1P_2}+Q_{1e}-\doublehat{Q}_{1e}+Q_{2e}-\doublehat{Q}_{2e}}{\doublehat{Q}_{1e}+\doublehat{Q}_{2e}+N_3}\right)-‎
\c_{12}-\c_{21}-\c_{13}-\c_{23}
\Bigg\},
\end{align}
where the union is taken over $\alpha_k,\alpha_{k3}\in\mathbb{R}^+$, $\eta_k \in [1-\min(1,Q_k/P_k),1]$, $\rho_k \in[0,1]$ and $P'_k,P''_k \geq 0$ such that $P'_k+P''_k\leq P_k$. For $k\in\{1,2\}$, we have
\begin{align*}
Q_{ke}&:=‎\left(‎‎‎\sqrt{Q_k}‎-‎\sqrt{‎\bar{\eta}_kP_k‎}‎\right)‎^2,~~ \hat{Q}_{ke}:=\tfrac{(1-\alpha_k)^2P'_{ke}Q_{ke}}{P'_{ke}+\alpha_k^2Q_{ke}},~~~
\doublehat{Q}_{ke}:=\tfrac{(1-\alpha_k-\alpha_{k3})^2P'_{ke}P''_{ke}Q_{ke}}{P'_{ke}P''_{ke}+(\alpha_k^2P''_{ke}+\alpha_{k3}^2P'_{ke})Q_{ke}},\\
\c_{k,3-k}&:=\c \left(\tfrac{\alpha^2_kQ_{ke}}{P'_{ke}}\right),~~~~~~~~~~~~~~
\c_{k3}:=\c \left(\tfrac{\alpha^2_{k3}\hat{Q}_{ke}}{(1-\alpha_k)^2P''_{ke}}\right),~~~
\delta_k:=\c \left(\tfrac{P'_{ke}+Q_{ke}-\hat{Q}_{ke}}{P''_{ke}+\hat{Q}_{ke}+\doublehat{Q}_{3-k,e}+N_3}\right)-\c_{k,3-k},\\
\Delta_k&:=\c \left(\tfrac{P'_{ke}+Q_{ke}-\hat{Q}_{ke}}{P''_{1e}+P''_{2e}+\hat{Q}_{1e}+\hat{Q}_{2e}+N_3}\right)-\c_{k,3-k},~~
\Delta_3:=\c \left(\tfrac{P'_{1e}+P'_{2e}+Q_{1e}-\hat{Q}_{1e}+Q_{2e}-\hat{Q}_{2e}}{P''_{1e}+P''_{2e}+\hat{Q}_{1e}+\hat{Q}_{2e}+N_3}\right)-\c_{12}-\c_{21}.
\end{align*}
}\end{proposition}
\emph{Proof of Proposition \ref{propos:proposition2}}: See Appendix III.
\begin{remark}\emph{Although the achievable rate region $\mathcal{R}^d_G$ seems to be complex, we have observed in most cases that the rate region is a convex combination of the following four cases;
\begin{itemize}
\item \textbf{Case 1} ($R_{13}=R_{23}=0$):
Both of encoders transmit their messages to the receiver by fully cooperating with each other. For this case, we have
\begin{equation*}
P''_1=P''_2=0,~~ \alpha_{13}=\alpha_{23}=0,~~\delta_1^-=\delta_2^-=\Delta^-=0.
\end{equation*}
\item \textbf{Case 2} ($R_{12}=R_{21}=0$):
Both of encoders ignore their feedback signals, and directly transmit their messages to the receiver without cooperation. For this case, we have
\begin{equation*}
P'_1=P'_2=0,~~\alpha_1=\alpha_2=0,~~ \rho_1=\rho_2=0, ~~\delta_1^-=\delta_2^-=\Delta^-=0.
\end{equation*}
\item \textbf{Case 3} ($R_{13}=R_{21}=0$):
Encoder 1 ignores the feedback signal, but encoder 2 relays the message of encoder 1 along with his private message to the receiver. For this case, we have
\begin{equation*}
P''_1=P'_2=0,~~\alpha_{13}=\alpha_2=0,~~\delta_1^-=\delta_2^-=\Delta^-=0.
\end{equation*}
\item \textbf{Case 4} ($R_{12}=R_{23}=0$):
Encoder 2 ignores the feedback signal, but encoder 1 relays the message of encoder 2 along with his private message to the receiver. For this case, we have
\begin{equation*}
P'_1=P''_2=0,~~\alpha_1=\alpha_{23}=0,~~\delta_1^-=\delta_2^-=\Delta^-=0.
\end{equation*}
\end{itemize}
}\end{remark}
In the following two special channel models are discussed.
\begin{enumerate}
  \item \emph{Doubly dirty Gaussian GMAC with informed receiver}:
  Assume that the receiver also knows the states $(S_0,S_1,S_2)$. Replacing $Y_3$ with $(Y_3,S_0,S_1,S_2)$ in Theorem \ref{th:theorem1}, an achievable rate region is derived. Moreover for the Gaussian channel model, the rate region of the clean Gaussian GMAC given in Proposition \ref{propos:proposition1}, equation (\ref{eq: proposition1}), is achieved.
  \item \emph{Dirty Gaussian GMAC with availability of full CSI at one encoder}:
  Let $S_0=S_2=\emptyset$. That is the channel is controlled by $S_1$ which is noncausally available only at encoder 1.
  \begin{proposition}\label{propos:proposition3}\emph{For the dirty Gaussian GMAC wherein $Q_0=Q_2=0$ and $Q_1\longrightarrow \infty$, the achievable rate region is
  \begin{align}\nonumber
  \mathcal{R}_G:= \bigcup \Bigg\{(R_1,R_2):~R_2  & \leq \c\Big(\frac{\bar{\rho_2}P'_2}{\bar{\rho_2}P''_2+N_1}\Big)+\c\Big(\frac{\bar{\rho_2}P''_2}{\big(\tfrac{1-\alpha_{13}}{\alpha_{13}}\big)^2\bar{\rho_1}P_1+N_3}\Big),\\
  R_1+R_2 & \leq \c \Big(-1+\frac{\bar{\rho_1}P_1}{(1-\alpha_{13})^2\bar{\rho_1}P_1+\alpha_{13}^2N_3}\Big)\Bigg\},
  \end{align}
  where the union is taken over $\rho_k\in[0,1]$, $\alpha_{13}\in \mathbb{R}^+$, and $P'_2+P''_2 \leq P_1$.}
\end{proposition}
\begin{proof}
It is sufficient to substitute the following parameters in Proposition \ref{propos:proposition2};
\vspace{-.5cm}
‎‎\begin{align*}‎‎
Q_2&=0,~\alpha_2=\alpha_{23}=0‎,~‎‎\eta_2=1,~\delta_2^-=0,\\
P'_1 &= 0,~\alpha_1=0,~\eta_1=1,~\delta_1^-=\Delta^-=0.
‎\end{align*}‎‎‎\end{proof}
\begin{remark}\emph{Maximum values of $R_1+R_2$ and $R_2$ in Proposition \ref{propos:proposition3} are;}\end{remark}

\begin{itemize}
\item The maximum sum-rate is
    \begin{equation*}
    \max (R_1+R_2) = \c\Big(\frac{P_1}{N_3}\Big).
    \end{equation*}
\item The maximum single rate of $R_2$ is
\begin{itemize}
\item For $N_1 \leq N_3$,
\begin{equation*}
\max R_2 = \min \Bigg\{\c\left(\frac{P_2}{N_1}\right),\c\left(\frac{P_1}{N_3}\right)\Bigg\}.
\end{equation*}
\item For $N_1 \geq N_3$,
\begin{align*}
\max R_2 =
\left\{
\begin{array}{rl}
&\c\left(\frac{P_2}{N_3} \right),~~~~~~~~~~~~~~~\text{ if } P_1 \geq P_2+N_3\\
&\c\left(\frac{P_2}{\big(\tfrac{1-\alpha_{13}^*}{\alpha_{13}^*}\big)^2P_1+N_3}\right),~ \text{ if } P_1 < P_2+N_3~:~ \alpha_{13}^*=\tfrac{2P_1}{P_1+P_2+N_3}.
\end{array} \right.
\end{align*}
\end{itemize}
\end{itemize}
\end{enumerate}
By numerical examples, Proposition \ref{propos:proposition2} and \ref{propos:proposition3} are discussed in the following.

Assume a doubly dirty Gaussian GMAC in which $S_1$ is noncausally available at encoder 1 and $S_2$ is noncausally available at encoder 2, where the channel parameters are set to $P_1=P_2=10\text{ [dB]}, N_1=N_2=0\text{ [dB]}$, $N_3=10\text{ [dB]}$ and $Q_1=Q_2\in\{7,13\}$ [dB]. For the given parameters, achievable rate regions are plotted for different scenarios in Fig. \ref{fig:numerical2}. It is shown that when each user partially cleans the interference and utilizes GDPC, how much the achievable rate region is enlarged compared to utilizing pure DPC (without partially cleaning the interferences). Moreover, achievable rate regions of the clean MAC and clean GMAC are plotted for comparison. As it is shown in Fig. \ref{fig:6a}, for moderate values of interferences' powers, cooperation between users is very beneficial. As the interferences become stronger, encoders cannot combat the interferences perfectly, as illustrated in Fig. \ref{fig:6b}.

%%%%%%%%%%%%%%%%%%%%%%%%%%%%%%%%%%%%%%%%%%%%%%%%%%%   Gaussian Numerical results 2
\begin{figure}[!t]
        \centering
        \begin{subfigure}[b]{0.5\textwidth}
                \centering
                \includegraphics[scale=0.5]{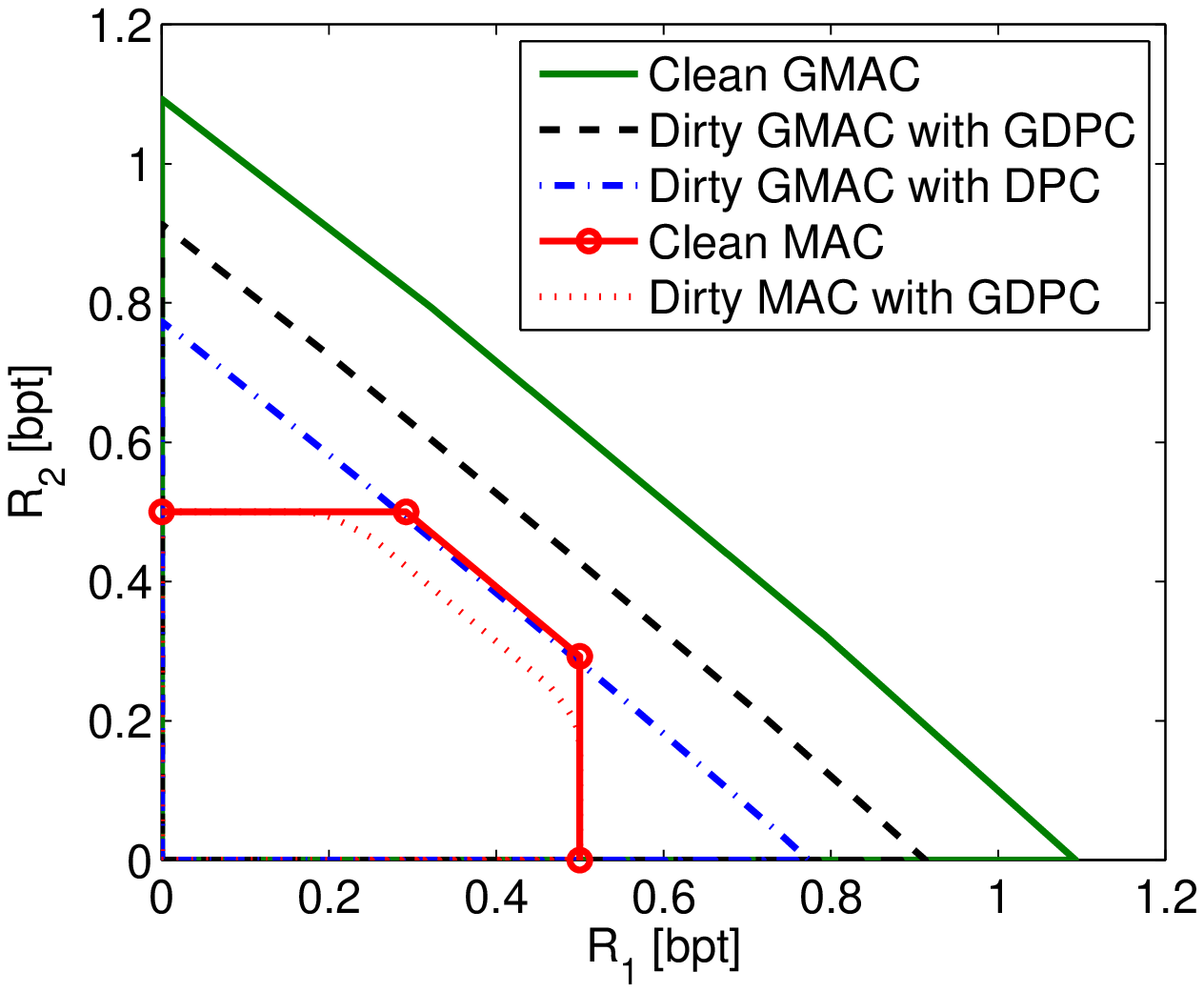}
                \caption{$Q_1=Q_2=7$ [dB]}
                \label{fig:6a}
        \end{subfigure}%
        \begin{subfigure}[b]{0.5\textwidth}
                \centering
                \includegraphics[scale=0.5]{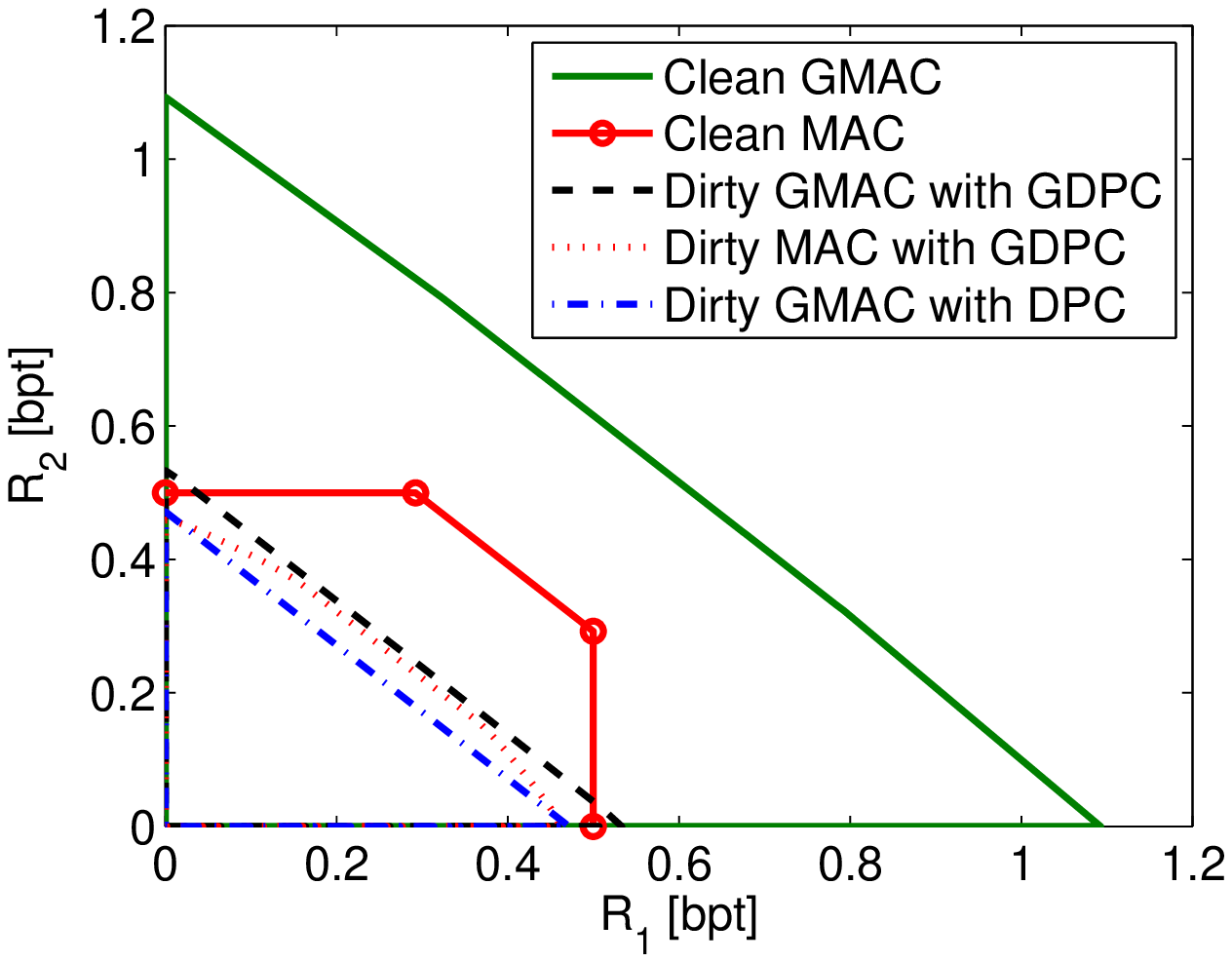}
                \caption{$Q_1=Q_2=13$ [dB]}
                \label{fig:6b}
        \end{subfigure}
        \caption{\small Achievable rate regions for the state-dependent GMAC with partial CSI at the encoders for given channel parameters $P_1=P_2=10\text{ [dB]}, N_1=N_2=0\text{ [dB]}$, $N_3=10\text{ [dB]}$.}
        \label{fig:numerical2}
        \vspace{-.5cm}
\end{figure}
%%%%%%%%%%%%%%%%%%%%%%%%%%%%%%%%%%%%%%%%%%%%%%%%%%%%%%%%%%%%%%%%%%%%%%%%%%%%%%%%%%%%%%%%%%%%%%%%%%%%%%
%%%%%%%%%%%%%%%%%%%%%%%%%%%%%%%%%%%%%% figure
\begin{figure}[!t]
\centering
\includegraphics[scale=.5]{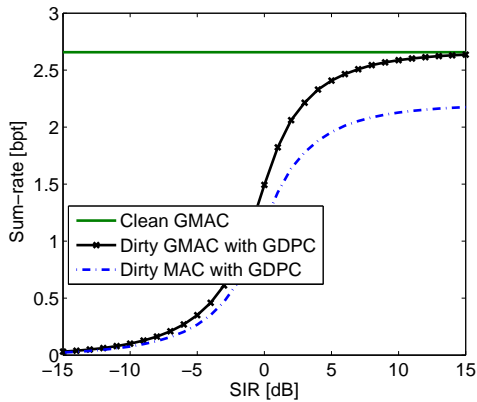}
\caption{\small Sum-rate versus SIR for $P_1=P_2=10$ [dB], $N_1=N_2=-10$ [dB] and $N_3=0$ [dB].}
\label{fig:sumratevsSIR}
\vspace{-.6cm}
\end{figure}
%%%%%%%%%%%%%%%%%%%%%%%%%%%%%%%%%%%%%%%%%%%%%%%%%%%%%%%%%%%%%%%%%%%%%%%%%%%%%%%%%%%%%%%%%%%%%%%%%%%%%%
In Fig. \ref{fig:sumratevsSIR}, sum-rate is plotted as a function of  signal-to-interference ratio (SIR) for $P_1=P_2=P$ [dB], $Q_1=Q_2=Q$ [dB]. Hence, SIR=$P-Q$ [dB]. An interesting observation is that whenever $N_1,N_2 \leq N_3$, independent of the interferences' powers, the encoders always prefer to fully cooperate with each other. For this case, the cooperative transmission outperforms the non-cooperative one (i.e. standard MAC).

In Fig. \ref{fig:numerical3}, achievable rate regions are plotted for the dirty Gaussian GMAC and channel parameters $S_2=\emptyset$, $Q_1 \longrightarrow \infty$, $P_2=10\text{ [dB]}, N_1=N_2=0\text{ [dB]}$, $N_3=10\text{ [dB]}$ and two values of $P_1\in \{10,15\}$ [dB]. As it is derived in Proposition \ref{propos:proposition3}, when at least one of the users has access to the full CSI noncausally, we can propose an achievable rate region which is independent of the interference signal. In Fig. \ref{fig:numerical3}, in addition to the regions of the clean MAC and GMAC, an outer bound on the capacity region of the dirty MAC (without cooperating encoders) is depicted, which is proposed in \cite{kotagiri2008multiaccess}. In Fig. \ref{fig:7a}, it is shown that the achievable rate region of dirty GMAC is equal to the outer region of dirty MAC, but maximum achievable $R_2$ for the dirty MAC is less than $\c (P_1/N_3)$. In Fig. \ref{fig:7b} for $P_2=15$ [dB], the achievable rate region of dirty GMAC increases dramatically. For this case, $R_1$ and $R_2$ approach $\c (P_1/N_3)$ and the achievable rate region of dirty MAC becomes close to the outer region.

%%%%%%%%%%%%%%%%%%%%%%%%%%%%%%%%%%%%%%%%%%%%%%%%%%%   Gaussian Numerical results 3
\begin{figure}[!t]
        \centering
        \begin{subfigure}[b]{0.5\textwidth}
                \centering
                \includegraphics[scale=0.5]{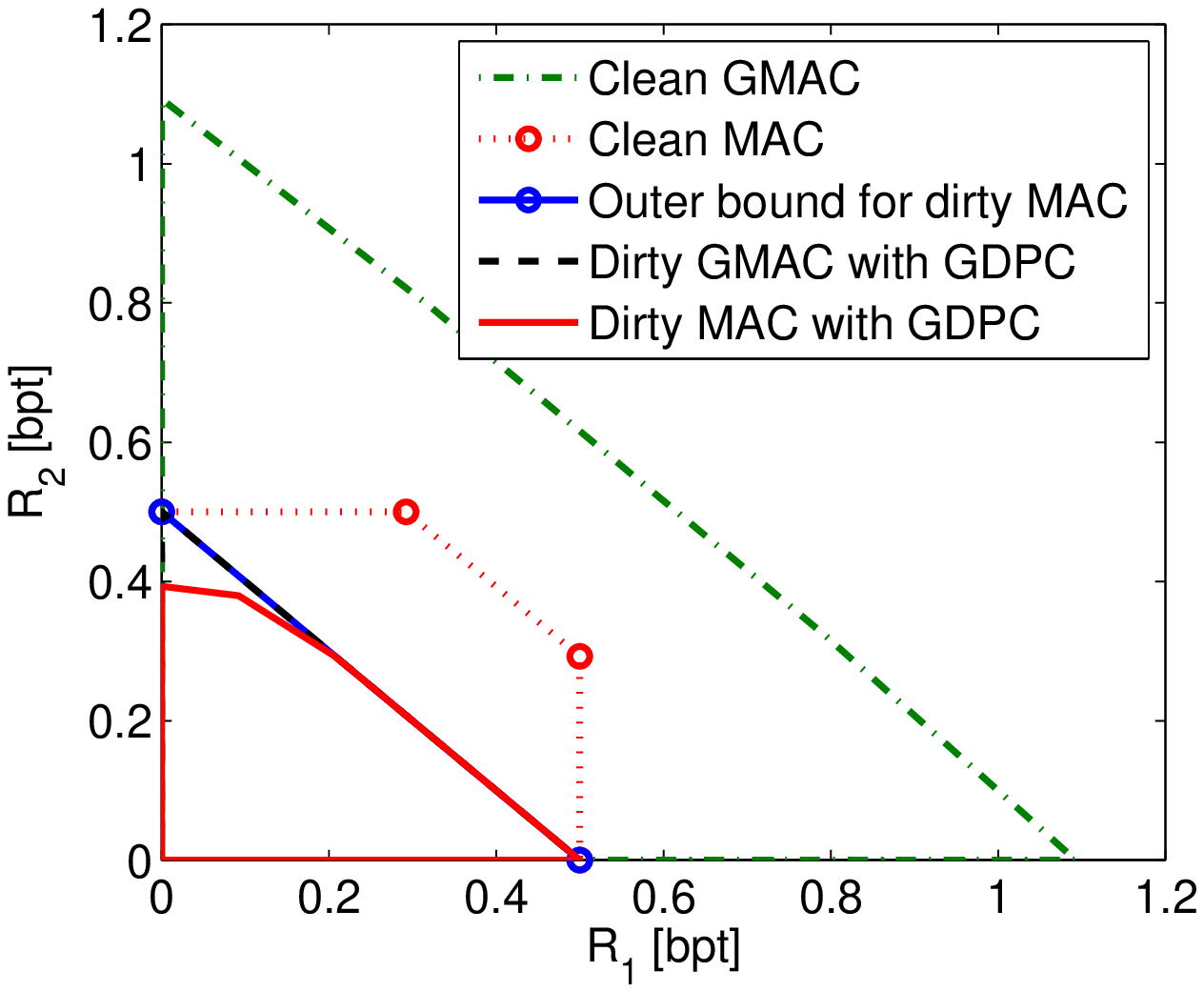}
                \caption{$P_1=10$ [dB]}
                \label{fig:7a}
        \end{subfigure}%
        \begin{subfigure}[b]{0.5\textwidth}
                \centering
                \includegraphics[scale=0.5]{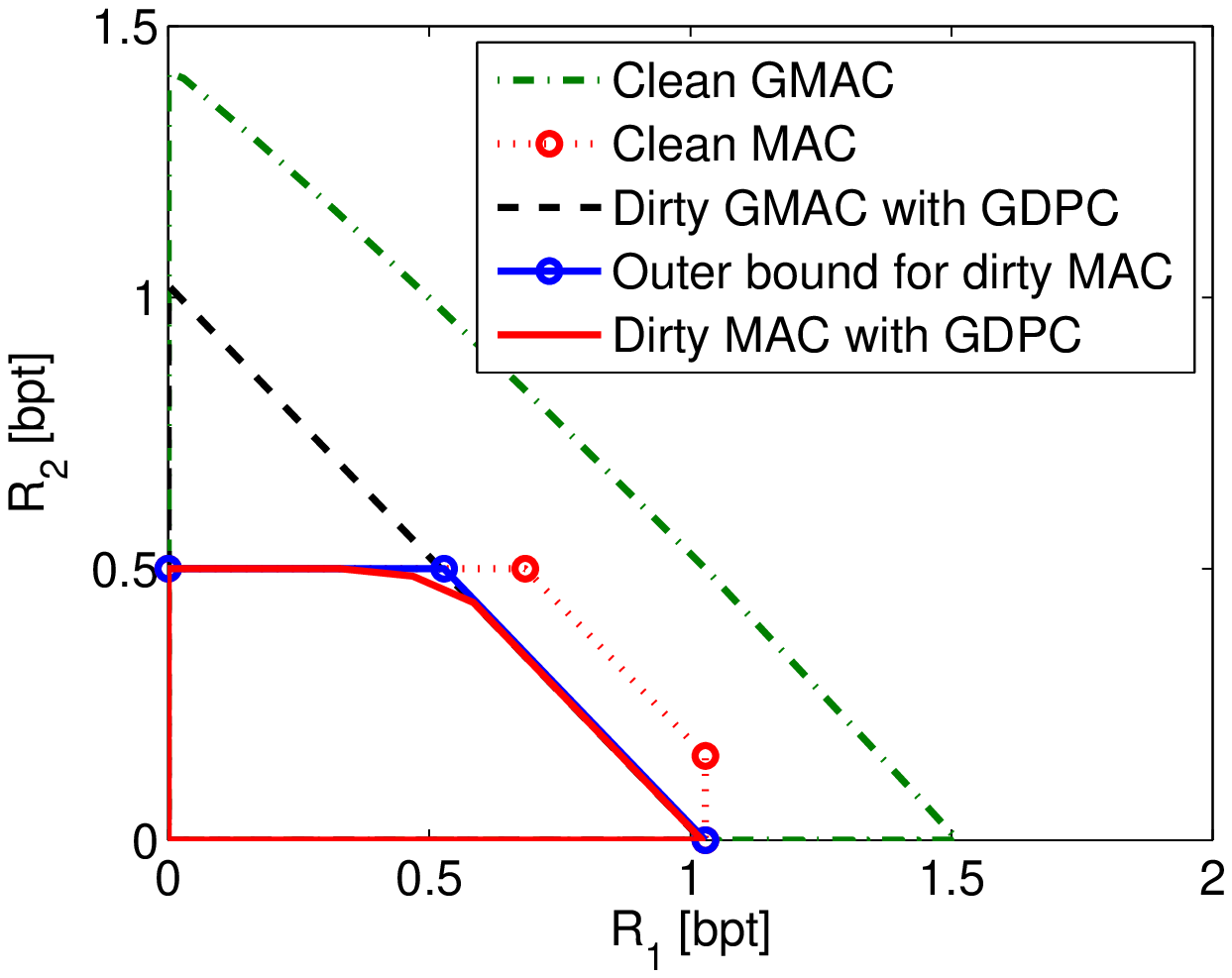}
                \caption{$P_1=15$ [dB]}
                \label{fig:7b}
        \end{subfigure}
        \caption{\small Achievable rate regions for the state-dependent GMAC with full CSI at the encoder 1 for given channel parameters $S_2=\emptyset$, $Q_1 \longrightarrow \infty $, $P_2=10\text{ [dB]}, N_1=N_2=0\text{ [dB]}$, $N_3=10\text{ [dB]}$.}
        \label{fig:numerical3}
        \vspace{-.5cm}
\end{figure}
%%%%%%%%%%%%%%%%%%%%%%%%%%%%%%%%%%%%%%%%%%%%%%%%%%%%%%%%%%%%%%%%%%%%%%%%%%%%%%%%%%%%%%%%%%%%%%%%%%%%%%
%%%%%%%%%%%%%%%%%%%%%%%%%%%%%%%%%%%%%%%%%%%%%%%%%%%%%%%%%%%%%%%%%%%%%%%%%%%%%%%%%%%%%%%%%%%%%%%%%%%%%%%%%%%%%%%%%%%%%%%%%%%%%%%%%%%%%%%%%%%%%%%%%
\section{Conclusions}
We investigated the state-dependent GMAC in which each encoder had
access to partial and noncausal CSIT. Since the encoders receive
feedback from the channel, they can cooperate to transmit coherently a
part of their messages to the receiver.  We constructed a coding
scheme based on rate-splitting, block-Markov encoding,
Gelfand--Pinsker multicoding and superposition coding techniques. To
enable the cooperation in our proposed strategy, each encoder
\emph{partially} decodes the message of the other encoder and the
receiver employs backward decoding with joint unique decoding at each
stage.
Afterwards, we tailored the proposed achievable rate region to a
Gaussian GMAC with an additive interference that is available
noncausally at both of the encoders. By an appropriate choice of
auxiliary random variables and dirty paper coding coefficients, the
effect of the interference is completely removed. Then, a doubly dirty
GMAC is considered and an achievable rate region is derived.
Although the proposed rate region is complex, we observed that that this region can be obtained as a
convex combination of four simple cases.
By numerical examples, it is shown that whenever the feedback links are less noisy than the direct links to the receiver, independent of the interferences' powers, the encoders benefit from cooperating with each other. Moreover, we proved that if at least one of the users knows the CSI completely, then an achievable rate region is established that is independent of the interference signal. Finally by presenting some numerical examples for this channel model, we showed a noticeable enlargement of the achievable rate region due to cooperation between the encoders.
%%%%%%%%%%%%%%%%%%%%%%%%%%%%%%%%%%%%%%%%%%%%%%%%%%%%%%%%%%%%%%%%%%%%%%%%%%%%%%%%%%%%%%%%%%%%%%%%%%%%%%%%%%%%%%%%%%%%%%%%%%%%%%%%%%%%%%%%%%%%%%%%%
%%%%%%%%%%%%%%%%%%%%%%%%%%%%%%%%%%%%%%%%    Appendix I
\section*{Appendix I}
To bound probability of error and derive the achievable rate region in
Theorem \ref{th:theorem1}, we use the union bound, packing lemma, AEP and Markov
lemma \cite{elgamal2011network}. In the following, by bounding
probability of error events at the encoders and the receiver, the rate
constraints given in Theorem \ref{th:theorem1} are established.

\emph{Analysis of probability of error}: The average error probability
for each block $b\in[1:B]$ is given by
\begin{eqnarray}
\label{eq:erranlyz}
\nonumber{\pr}_e&\leq&\sum_{(\bs_0,\bs_1,\bs_2)\notin A_\epsilon^n(S_0,S_1,S_2)}p(\bs_0,\bs_1,\bs_2)
+\sum_{(\bs_0,\bs_1,\bs_2)\in A_\epsilon^n (S_0,S_1,S_2)}p(\bs_0,\bs_1,\bs_2)\times \pr(\text{Error}|\bs_0,\bs_1,\bs_2)\\
&_{\leq}^{(a)}& \varepsilon+\sum_{(\bs_0,\bs_1,\bs_2)\in A_\epsilon^n (S_0,S_1,S_2)}p(\bs_0,\bs_1,\bs_2)\times \pr(\text{Error}|\bs_0,\bs_1,\bs_2),
\end{eqnarray}
where $(a)$ follows by AEP as $n \rightarrow \infty$
\cite{elgamal2011network}. To bound the second term, we need to
analyze all error events that may occur at the encoders and the
receiver for $b\in[1:B]$. First, we analyze error events at the
encoders and then error events at the receiver are considered.
\begin{itemize}
\item \emph{Error events due to the decoding at encoder 1}:\\
Suppose that encoder 2 has transmitted $M^b_{21}$ over block $b$. Encoder 1 receives $\by^b_1$, and tries to decode $M^b_{21}$ given $(\bu(\doublehat{m}_c^{b},j^*_0),\bs^b_1,\bx^b_1)$ by joint typicality decoding given in (\ref{eq: decodingEn1}). Thus, the following error events are defined.
\begin{align*}
  &
  E_{1,1}^b:=\Big\{\big(\bu(\doublehat{m}_c^{b},j^*_0),\bv_2(\doublehat{m}^b_c,j^*_{u_0},M^b_{21},j_2),\by_1^b,\bs_0^b,\bs^b_1,\bx^b_1\big)
  \notin A_\epsilon^n (U,V_2,Y_1,S_0,S_1,X_1)\Big\}\\\nonumber &
  E_{1,2}^b:=\Big\{\exists m^b_{21}\neq M^b_{21}
  :\big(\bu(\doublehat{m}_c^{b},j^*_0),\bv_2(\doublehat{m}^b_c,j^*_{u_0},m^b_{21},j_2),\by_1^b,\bs_0^b,\bs^b_1,\bx^b_1\big)\in
  A_\epsilon^n (U,V_2,Y_1,S_0,S_1,X_1)\Big\}.
\end{align*}
By Markov lemma \cite{elgamal2011network}, $P(E^b_{1,1})$ tends to zero as $n\rightarrow\infty$.
Next, utilizing the union bound and joint typicality lemma, the error event is bounded as
\begin{eqnarray}\label{eq: R21}
\pr\big(E_{1,2}^b|\overline{E_{1,1}^b}\big)&\leq& 2^{nR_{21}}J_2 \times 2^{-n[I(V_2;S_0|U)+I(V_2;Y_1|S_0S_1UX_1)-\epsilon]},
\end{eqnarray}
where $\overline{E_{1,1}^b}$ denotes the complement of $E_{1,1}^b$. Substituting $J_2$ in (\ref{eq: R21}), the probability of error in (\ref{eq: R21}) tends to zero for sufficiently large $n$, if
\begin{eqnarray}\label{eq: R21Final}
R_{21}&<& I(V_2;Y_1|S_0S_1UX_1)-I(V_2;S_2|S_0U).
\end{eqnarray}
\item \emph{Error events due to the decoding at encoder 2}:\\
With a similar discussion to encoder 1, error events at encoder 2 can be defined. The probability of error tends to zero for sufficiently large $n$, if
\begin{eqnarray}\label{eq: R12Final}
R_{12}&<& I(V_1;Y_2|S_0S_2UX_2)-I(V_1;S_1|S_0U).
\end{eqnarray}
\item \emph{Error events due to the decoding at the receiver}:\\
We analyze the decoding error events for block $b$. It is assumed that the encoders have transmitted $(M^b_c,M^b_{12},M^b_{21},$ $M^b_{13},M^b_{23})$ over block $b$. Since the receiver uses the backward decoding, it is also assumed that the receiver already knows $\hat{M}^{b+1}_c=(\hat{M}^b_{12},\hat{M}^b_{21})$ from decoding of $\by_3^{b+1}$. Therefore, using $\by_3^b$, the receiver tries to estimate $(M^b_c,M^b_{13},M^b_{23})$ form the joint typicality of sequences given in (\ref{eq: RXdecoding}). Thus, we have 11 error events to consider as summarized in Table~\ref{tabel: errorevents}.
\begin{table*}[!t]\footnotesize
\caption{\small Illustration of error events at the receiver. Checkmark denotes correctness of the message or index, and cross denotes error.}\label{tabel: errorevents}
\begin{center}
\begin{tabular}{ |c || c| c | c| c|c|c|c|}\hline
&$m_c$& $j_0$ & $j_1$ & $j_2$ & $m_{13}$ & $m_{23}$ \\ \hline\hline
$E^b_{3,1}$& $\checkmark$& $\checkmark$ & $\checkmark$ & $\checkmark$ & $\checkmark$ & $\checkmark$ \\ \hline
$E^b_{3,2}$& $\times$& $\times$ & $\times$ & $\times$ & $\times$ & $\times$ \\ \hline
$E^b_{3,3}$& $\checkmark$& $\times$ & $\times$ & $\times$ & $\times$ & $\times$ \\ \hline
$E^b_{3,4}$& $\checkmark$& $\checkmark$ & $\times$ & $\times$ & $\times$ & $\times$ \\ \hline
$E^b_{3,5}$& $\checkmark$& $\checkmark$ & $\checkmark$ & $\times$ & $\times$ & $\times$ \\ \hline
$E^b_{3,6}$& $\checkmark$& $\checkmark$ & $\times$ & $\checkmark$ & $\times$ & $\times$ \\ \hline
$E^b_{3,7}$& $\checkmark$& $\checkmark$ & $\checkmark$ & $\checkmark$ & $\times$ & $\times$ \\ \hline
$E^b_{3,8}$& $\checkmark$& $\checkmark$ & $\checkmark$ & $\checkmark$ & $\checkmark$ & $\times$ \\ \hline
$E^b_{3,9}$& $\checkmark$& $\checkmark$ & $\checkmark$ & $\checkmark$ & $\times$ & $\checkmark$ \\ \hline
$E^b_{3,10}$& $\checkmark$& $\checkmark$ & $\checkmark$ & $\times$ & $\checkmark$ & $\times$ \\ \hline
$E^b_{3,11}$& $\checkmark$& $\checkmark$ & $\times$ & $\checkmark$ & $\times$ & $\checkmark$ \\ \hline
\end{tabular}
\end{center}
\vspace{-.5cm}
\end{table*}
\begin{remark}\emph{Since multi-layer GPC technique is used for the codeword generations, to decode $m_{k3}$ correctly, we need to decode $j_0$ and $j_k$ for $k\in\{1,2\}$. Thus, these error events are also considered in the analysis of error events at the receiver.}
\end{remark}
Various error events which needs to be bounded are summarized in Table~\ref{tabel: errorevents}. First, let
\begin{align*}
E_{3,1}^b&:=\{E^b_D|_{M^b_c,M^b_{13},M^b_{23}} \notin A_\epsilon^n (D)\}\\
E_{3,2}^b&:=\{\exists m_c^b\neq M^b_c: E^b_D \in A_\epsilon^n (D)\}\\
E_{3,3}^b&:=\{\exists j_0\neq j^*_0: E^b_D|_{M_c^b} \in A_\epsilon^n (D)\}.
\end{align*}
By use of Markov lemma \cite{elgamal2011network}, $P(E^b_{3,1})$ tends to zero as $n\rightarrow\infty$. By utilizing joint typicality lemma and union bound, we have
\begin{subequations}\label{eq: Rs}
\begin{align}
\pr\big(E_{3,2}^b|\overline{E_{3,1}^b} \big)&\leq
2^{n(R_{12}+R_{21}+R_{13}+R_{23})}J_0J_1J_2J_{13}J_{23}\times 2^{-n[I(UV_1V_2V_{13}V_{23};Y_3)+I(V_1V_{13};V_2V_{23}|U)-\epsilon]}\\
\pr\big(E_{3,3}^b|\overline{E_{3,1}^b} \big)&\leq
2^{n(R_{13}+R_{23})}J_0J_1J_2J_{13}J_{23} \times 2^{-n[I(UV_1V_2V_{13}V_{23};Y_3)+I(V_1V_{13};V_2V_{23}|U_0)-\epsilon]}.
\end{align}
\end{subequations}
Since $0\leq R_{12}+R_{21}$, the rate constraint obtained via (\ref{eq: Rs}b) becomes redundant as compared to that found via (\ref{eq: Rs}a). Substituting $J_0,J_k,J_{k3}$ for $k\in\{1,2\}$ in (\ref{eq: Rs}a), we find the following sum-rate
\begin{eqnarray}\label{eq: RsFinal}\nonumber
R_{12}+R_{21}+R_{13}+R_{23} &<& I(UV_1V_2V_{13}V_{23};Y_3)-I(UV_1V_2V_{13}V_{23};S_0)-\\
&& I(V_1V_{13};S_1|S_0U)-I(V_2V_{23};S_2|S_0U).
\end{eqnarray}
Now define the following error events
\begin{align*}
E_{3,4}^b&:=\{\exists j_1\neq j^*_1 \text{ and } j_2 \neq j^*_2: E^b_D|_{M_c^b,j^*_0} \in A_\epsilon^n (D)\}\\\nonumber
E_{3,5}^b&:=\{\exists j_2\neq j^*_2 \text{ and } m^b_{13} \neq M^b_{13}: E^b_D|_{M_c^b,j^*_0,j^*_1} \in A_\epsilon^n (D)\}\\\nonumber
E_{3,6}^b&:=\{\exists j_1\neq j^*_1 \text{ and } m^b_{23} \neq M^b_{23}: E^b_D|_{M_c^b,j^*_0,j^*_2} \in A_\epsilon^n (D)\}
\end{align*}
Again by applying union bound technique and joint typicality lemma, the following probability of errors are bounded as
\begin{subequations}\label{eq: Rj}
\begin{align}
\pr\big(E_{3,4}^b|\bigcap_{m=1}^3 \overline{E_{3,m}^b} \big)&\leq
2^{n(R_{13}+R_{23})}J_1J_2J_{13}J_{23}\times2^{-n[I(V_1,V_2,V_{13},V_{23};Y_3|U)+I(V_2V_{23};V_1V_{13}|U)-\epsilon]}\\
\pr\big(E_{3,5}^b|\bigcap_{m=1}^3 \overline{E_{3,m}^b} \big)&\leq
2^{n(R_{13}+R_{23})}J_2J_{13}J_{23}\times2^{-n[I(V_2,V_{13},V_{23};Y_3|UV_1)+I(V_2V_{23};V_1V_{13}|U)-\epsilon]}\\
\pr\big(E_{3,6}^b|\bigcap_{m=1}^3 \overline{E_{3,m}^b} \big)&\leq
2^{n(R_{13}+R_{23})}J_1J_{13}J_{23}\times2^{-n[I(V_1,V_{13},V_{23};Y_3|UV_2)+I(V_1V_{13};V_2V_{23}|U)-\epsilon]}.
\end{align}
\end{subequations}
Plugging in $J_k$ and $J_{k3}$ for $k\in\{1,2\}$ in (\ref{eq: Rj}), subject to the following rate constraints, the probability of errors in (\ref{eq: Rj}) tend to zero for sufficiently large $n$.
\begin{subequations}\label{eq: Rjfinal}
\begin{eqnarray}\nonumber
R_{13}+R_{23} &<& I(V_1V_2V_{13}V_{23};Y_3|U)-I(V_1V_2V_{13}V_{23};S_0|U)-\\
&& I(V_1V_{13};S_1|S_0U)-I(V_2V_{23};S_2|S_0U)\\\nonumber
R_{13}+R_{23} &<& I(V_2V_{13}V_{23};Y_3|UV_1)-I(V_2V_{13}V_{23};S_0|UV_1)-\\
&& I(V_{13};S_1|S_0UV_1)-I(V_2V_{23};S_2|S_0U)\\\nonumber
R_{13}+R_{23} &<& I(V_1V_{13}V_{23};Y_3|UV_2)-I(V_1V_{13}V_{23};S_0|UV_2)-\\
&& I(V_1V_{13};S_1|S_0U)-I(V_{23};S_2|S_0UV_2).
\end{eqnarray}
\end{subequations}

\end{itemize}
Following a same steps, other error events are bounded. Finally by use of the following remarks, the derived rate constraints are simplified and Theorem 1 is proved.
%
%We next simplify the obtained  rate constraints in (\ref{eq: RsFinal}), (\ref{eq: Rjfinal}), (\ref{eq: Rk3final}) and (\ref{eq: R13orR23Final}), using the following observation.
\begin{remark}\emph{
From the given pmf in (\ref{eq: pmf}), the following Markov chains are valid.}
\begin{subequations}\label{eq: markovchain}
\begin{eqnarray}
& S_1S_2 \longleftrightarrow S_0 \longleftrightarrow U, ~~~
V_1V_{13} \longleftrightarrow S_0U \longleftrightarrow V_2V_{23}.
\end{eqnarray}
\end{subequations}
\end{remark}
\vspace{-0.5cm}
\section*{Appendix II}
In the following, we derive the optimal precoding coefficients. Note that, if we could completely remove the interference over the channel using the proposed multi-layer Costa precoding in (\ref{eq:txsignals}) and (\ref{eq: x_k}), then the knowledge of the interference at the destination would \emph{not} increase the achievable rate region. Thus the rate region with optimal precoding should equal that when $S_0$ is also known at the destination. That is the following equalities should be satisfied for $k \in \{1,2\}$.
\begin{subequations}\label{eq: lemma1}
\begin{align}
%I(V_{13};Y_3|UV_1V_2V_{23})&=I(V_{13};Y_3S_0|UV_1V_2V_{23}),\\
%I(V_{23};Y_3|UV_1V_2V_{13})&=I(V_{23};Y_3S_0|UV_1V_2V_{13}),\\
I(V_{k3};Y_3|UV_1V_2V_{3-k,3})&=I(V_{k3};Y_3S_0|UV_1V_2V_{3-k,3}),\\
I(V_{13}V_{23};Y_3|UV_1V_2)&=I(V_{13}V_{23};Y_3S_0|UV_1V_2),\\
I(UV_1V_2V_{13}V_{23};Y_3)&=I(UV_1V_2V_{13}V_{23};Y_3S_0).
\end{align}
\end{subequations}
We next show that there exist precoding coefficients such that (\ref{eq: lemma1}d) holds. Thus, it is sufficient to prove;
\begin{equation}\label{eq: hproof1}
h(UV_1V_2V_{13}V_{23}|Y_3)=h(UV_1V_2V_{13}V_{23}|Y_3S_0).
\end{equation}
Next consider the left hand side of (\ref{eq: hproof1})
\begin{eqnarray}\label{eq: hproof2}\nonumber
h\big(UV_1V_2V_{13}V_{23}|Y_3\big)&=& h\big(X_u+\alpha_0S_0,\sqrt{\rho_1P_1}X_u+X'_1+\alpha_1S_0,\sqrt{\rho_2P_2}X_u+X'_2+\alpha_2S_0,X''_1+
\alpha_{13}S_0,\\\nonumber
&&X''_2+\alpha_{23}S_0\big|X_1+X_2+Z_3+S_0\big)\\
&=&h\big(\mathcal{E}_{U},\mathcal{E}_{V_1},\mathcal{E}_{V_2},\mathcal{E}_{V_{13}},\mathcal{E}_{V_{23}}|X_1+X_2+Z_3+S_0\big),
\end{eqnarray}
where, for $k\in\{1,2\}$
\begin{subequations}\label{eq: phi}
\begin{align}
\mathcal{E}_{U}&:=X_u-\alpha_0\big(X_1+X_2+Z_3\big),\\
 \mathcal{E}_{V_k}&:=\sqrt{\rho_kP_k}X_u+X'_k-\alpha_k(X_1+X_2+Z_3\big),\\
%\mathcal{E}_{V_2}&:=\sqrt{\rho_2P_2}X_u+X'_2-\alpha_2(X_1+X_2+Z_3\big),\\
 \mathcal{E}_{V_{k3}}&:=X''_k-\alpha_{k3}(X_1+X_2+Z_3\big),
%\mathcal{E}_{V_{23}}&:=X''_2-\alpha_{23}(X_1+X_2+Z_3\big),
\end{align}
\end{subequations}
and $X_k=\sqrt{\rho_k P_k}X_u+X'_k+X''_k$ for $k\in\{1,2\}$. Now choose the precoding coefficients such that the second norm of the random variables defined in (\ref{eq: phi}) is minimized. This yields (\ref{eq: DPCcoeff}a)--(\ref{eq: DPCcoeff}c). Note that with these optimal coefficients we have
\begin{subequations}\label{eq: expectetion}
\begin{align}
\mathbb{E}\{\mathcal{E}_{U}\times (X_1+X_2+Z_3)\}&=0,~~
\mathbb{E}\{\mathcal{E}_{V_k}\times (X_1+X_2+Z_3)\}=0,~~
\mathbb{E}\{\mathcal{E}_{V_{k3}}\times (X_1+X_2+Z_3)\}=0.
\end{align}
\end{subequations}
Finally observe that the random variables in (\ref{eq: phi}) are independent of $S_0$. Therefore, (\ref{eq: hproof2}) is equal to
\begin{eqnarray}\label{eq: hproofinal}
&&h(UV_1V_2V_{13}V_{23}|Y_3)=h(\mathcal{E}_{U},\mathcal{E}_{V_1},\mathcal{E}_{V_2},\mathcal{E}_{V_{13}},\mathcal{E}_{V_{23}})=
h(\mathcal{E}_{U},\mathcal{E}_{V_1},\mathcal{E}_{V_2},\mathcal{E}_{V_{13}},\mathcal{E}_{V_{23}}|Y_3,S_0).
\end{eqnarray}
This completes the proof of (\ref{eq: hproof1}) or equivalently (\ref{eq: lemma1}d). With a similar discussions (\ref{eq: lemma1}a)--(\ref{eq: lemma1}c) hold.
%%%%%%%%%%%%%%%%%%%%%%%%%%%%%%%%%%%%%%%%%%%%%%%%%%%%%%%%%%%%%%%%%%%%%%%%%%%%%%%%%%%%%%%%%%%%%%%%%%%%%%%%%%%%%%%%%%%%%%%%%%%%%%%%%%%
%%%%%%%%%%%%%%%%%%%%%%%%%%%%%%%%%%%%%%%%%%%%%%%%%%%%%%%%%%%%%%%%%%%%%%%%%%%%%%%%%%%%%%%%%%%%%%%%%%%%%%%%%%%%%%%%%%%%%%%%%%%%%%%%%%%
%%%%%%%%%%%%%%%%%%%%%%%%%%%%%%%%%%%%%%%%%%%%%%%%%%%%%%%%%%%%%%%%%%%%%%%%%%%%%%%%%%%%%%%%%%%%%%%%%%%%%%%%%%%%%%%%%%%%%%%%%%%%%%%%%%%
\vspace{-0.5cm}
\section*{Appendix III}
\begin{proof}[Proof of Proposition \ref{propos:proposition2}] To prove the proposition, it is sufficient to substitute the auxiliary random variables defined in (\ref{eq: AuxRvDGMAC}) into the achievable rate region given in Theorem \ref{th:theorem1}. For instance, we evaluate the term $h(Y_3|UV_1V_2V_{23})$ which is used in $R_{13}$, see Theorem \ref{th:theorem1}.
\begin{itemize}
\item Plug in $X_k$ form (\ref{eq: AuxRvDGMAC}c) into (\ref{eq:gaussianmodel}), we have
\begin{equation}
Y_3=(\sqrt{\eta_1\rho_1P_1}+\sqrt{\eta_2\rho_2P_2})^2U+X'_1+X''_1+X'_2+X''_2+\beta_1 S_1+\beta_2S_2+Z_3
\end{equation}
\item Since $U$ is known, $Y_3|U,V_1,V_2,V_{23}$ is equivalent to
\begin{align}\label{eq:14}\nonumber
&X'_1+X''_1+X'_2+X''_2+\beta_1 S_1+\beta_2S_2+Z_3 |U,X'_1+\alpha_1 \beta_1 S_1,X'_2+\alpha_2 \beta_2 S_2, X''_2+\alpha_{23} \beta_2 S_2\\
\equiv& X''_1+\beta_1(1-\alpha_1)S_1+\beta_2(1-\alpha_2-\alpha_{23})S_2+Z_3|X'_1+\alpha_1 \beta_1 S_1,X'_2+\alpha_2 \beta_2 S_2, X''_2+\alpha_{23} \beta_2 S_2
\end{align}
As $U$ is independent form other random variables, it is dropped from (\ref{eq:14}).
\item By use of minimum mean square error (MMSE) estimator, $S_1$ is estimated from $X'_1+\alpha_1 \beta_1 S_1$, and $S_2$ is estimated from $(X'_2+\alpha_2 \beta_2 S_2, X''_2+\alpha_{23} \beta_2 S_2)$. Therefore, $Y_3|U,V_1,V_2,V_{23}$ is equivalent to
\begin{align}\label{eq:15}
X''_1+\beta_1(1-\alpha_1)(S_1-\hat{S}_1)+\beta_2(1-\alpha_2-\alpha_{23})(S_2-\doublehat{S}_2)+Z_3,
\end{align}
where
\begin{align*}
&\mathbb{E} \{\beta_1^2(1-\alpha_1)^2(S_1-\hat{S}_1)^2\} = \hat{Q}_{1e},~~~~~
\mathbb{E} \{\beta_2^2(1-\alpha_2-\alpha_{23})^2(S_2-\doublehat{S}_2)^2\} = \doublehat{Q}_{2e}.
\end{align*}
Therefore,
\begin{equation}
h(Y_3|U,V_1,V_2,V_{23})=\tfrac{1}{2}\log_2\Big(2\pi e \big(P''_{1e}+\hat{Q}_{1e}+\doublehat{Q}_{2e}+N_3)\Big).
\end{equation}
With a similar discussions, Proposition \ref{propos:proposition2} is derived.
\end{itemize}
\end{proof}
\vspace{-1cm}
\bibliographystyle{ieeetr}
\bibliography{IEEEabrv,ref}

\begin{thebibliography}{10}

\bibitem{carleial1982multiple}
A.~Carleial, ``Multiple-access channels with different generalized feedback
  signals,'' {\em IEEE Transactions on Information Theory}, vol.~28,
  pp.~841--850, Nov. 1982.

\bibitem{willemsThesis}
F.~M.~J. Willems, {\em Information Theoretical Results for the Discrete
  Memoryless Multiple Access Channel}.
\newblock PhD thesis, Leuven, Belgium: Katholieke Universiteit, Oct. 1982.

\bibitem{Zeng1989IT}
C.-M. Zeng, F.~Kuhlmann, and A.~Buzo, ``Achievability proof of some multiuser
  channel coding theorems using backward decoding,'' {\em IEEE Transactions on
  Information Theory}, vol.~35, pp.~1160--1165, Nov. 1989.

\bibitem{Willems85}
F.~Willems and E.~van~der Meulen, ``The discrete memoryless multiple-access
  channel with cribbing encoders,'' {\em IEEE Transactions on Information
  Theory}, vol.~31, pp.~313--327, May 1985.

\bibitem{cover79}
T.~Cover and A.~El~Gamal, ``Capacity theorems for the relay channel,'' {\em
  IEEE Transactions on Information Theory}, vol.~25, pp.~572--584, Sept. 1979.

\bibitem{shannon1958channels}
C.~Shannon, ``Channels with side information at the transmitter,'' {\em IBM
  journal of Research and Development}, vol.~2, no.~4, pp.~289--293, 1958.

\bibitem{gel1980coding}
S.~Gel’fand and M.~Pinsker, ``Coding for channel with random parameters,''
  {\em Problems of control and information theory}, vol.~9, no.~1, pp.~19--31,
  1980.

\bibitem{costa1983writing}
M.~Costa, ``Writing on dirty paper (corresp.),'' {\em IEEE Transactions on
  Information Theory}, vol.~29, no.~3, pp.~439--441, 1983.

\bibitem{kotagiri2005reversible}
S.~Kotagiri and J.~N. Laneman, ``Reversible information embedding in multi-user
  channels,'' in {\em Proc. Allerton Conf. Communications, Control, and
  Computing}, Citeseer, 2005.

\bibitem{zaidi2007broadcast}
A.~Zaidi, P.~Piantanida, and P.~Duhamel, ``Broadcast-and mac-aware coding
  strategies for multiple user information embedding,'' {\em IEEE Transactions
  on Signal Processing}, vol.~55, no.~6, pp.~2974--2992, 2007.

\bibitem{keshet2007channel}
G.~Keshet, Y.~Steinberg, and N.~Merhav, ``Channel coding in the presence of
  side information,'' {\em Foundations and Trends in Communications and
  Information Theory}, vol.~4, no.~6, pp.~445--586, 2007.

\bibitem{pablothesis2007}
P.~Piantanida, {\em Multi-User Information Theory: State Information and
  Imperfect Channel Knowledge}.
\newblock PhD thesis, Paris-Sud University, Gif-sur-Yvette, France, 2007.

\bibitem{somekh2008cooperative}
A.~Somekh-Baruch, S.~Shamai, and S.~Verd{\'u}, ``Cooperative multiple-access
  encoding with states available at one transmitter,'' {\em IEEE Transactions
  on Information Theory}, vol.~54, pp.~4448--4469, Oct. 2008.

\bibitem{kotagiri2008multiaccess}
S.~Kotagiri and J.~Laneman, ``Multiaccess channels with state known to some
  encoders and independent messages,'' {\em EURASIP Journal on Wireless
  Communications and Networking}, vol.~2008, pp.~1--14, Jan. 2008.

\bibitem{KhosraviITW2011}
R.~Khosravi-Farsani and F.~Marvasti, ``Capacity bounds for multiuser channels
  with non-causal channel state information at the transmitters,'' in {\em IEEE
  Information Theory Workshop (ITW)}, pp.~195--199, Oct. 2011.

\bibitem{EmadiITW2012}
M.~F. Pourbabaee, M.~J. Emadi, A.~Gholami, and M.~R. Aref, ``Lattice coding for
  multiple access channels with common message and additive interference,'' in
  {\em Information Theory Workshop (ITW), 2012 IEEE}, Sept. 2012.

\bibitem{lapidoth2010multiple}
A.~Lapidoth and Y.~Steinberg, ``The multiple access channel with causal and
  strictly causal side information at the encoders,'' in {\em Proceedings of
  International Zurich Seminar on Communications}, pp.~13--16, 2010.

\bibitem{li2010multiple}
M.~Li, O.~Simeone, and A.~Yener, ``Multiple access channels with states
  causally known at transmitters,'' {\em IEEE Transactions on Information
  Theory}, vol.~59, no.~3, pp.~1394--1404, 2013.

\bibitem{zaidi2012capacity}
A.~Zaidi, P.~Piantanida, and S.~Shamai, ``Capacity region of multiple access
  channel with states known noncausally at one encoder and only strictly
  causally at the other encoder,'' {\em arXiv preprint arXiv:1201.3278}, 2012.

\bibitem{EmadiIET2012}
M.~Emadi, M.~Zamanighomi, and M.~Aref, ``Multiple-access channel with
  correlated states and cooperating encoders,'' {\em IET Communications},
  vol.~6, pp.~1857--1867, Sept. 2012.

\bibitem{Bross10}
S.~Bross and A.~Lapidoth, ``The state-dependent multiple-access channel with
  states available at a cribbing encoder,'' in {\em 26th IEEE Convention of
  Electrical and Electronics Engineers in Israel (IEEEI)}, pp.~665--669, Nov.
  2010.

\bibitem{abdellatif2009lower}
A.~Zaidi and L.~Vandendorpe, ``Lower bounds on the capacity of the relay
  channel with states at the source,'' {\em EURASIP Journal on Wireless
  Communications and Networking}, vol.~2009, pp.~1--22, Aug. 2009.

\bibitem{zaidi2010cooperative}
A.~Zaidi, S.~Kotagiri, J.~Laneman, and L.~Vandendorpe, ``Cooperative relaying
  with state available noncausally at the relay,'' {\em IEEE Transactions on
  Information Theory}, vol.~56, pp.~2272--2298, May 2010.

\bibitem{ZaidiIT2012}
A.~Zaidi, S.~Shamai, P.~Piantanida, and L.~Vandendorpe, ``Bounds on the
  capacity of the relay channel with noncausal state at the source,'' {\em
  \emph{to appear in} IEEE Trans. Inf. Theory}, 2013.

\bibitem{nasiri2013state}
M.~Khormuji, A.~El~Gamal, and M.~Skoglund, ``State-dependent relay channel:
  Achievable rate and capacity of a semideterministic class,'' {\em IEEE
  Transactions on Information Theory}, vol.~59, no.~5, pp.~2629--2638, 2013.

\bibitem{ElgamalAref82}
A.~El~Gamal and M.~Aref, ``The capacity of the semideterministic relay channel
  (corresp.),'' {\em IEEE Transactions on Information Theory}, vol.~28, p.~536,
  May 1982.

\bibitem{PhilosofIT09}
T.~Philosof and R.~Zamir, ``On the loss of single-letter characterization: The
  dirty multiple access channel,'' {\em IEEE Transactions on Information
  Theory}, vol.~55, pp.~2442--2454, June 2009.

\bibitem{kramer2007topics}
G.~Kramer, ``Topics in multi-user information theory,'' {\em Foundations and
  Trends in Communications and Information Theory}, vol.~4, no.~4-5,
  pp.~265--444, 2007.

\bibitem{ZahediIT2006}
A.~El~Gamal, M.~Mohseni, and S.~Zahedi, ``Bounds on capacity and minimum
  energy-per-bit for awgn relay channels,'' {\em IEEE Transactions on
  Information Theory}, vol.~52, pp.~1545--1561, April 2006.

\bibitem{elgamal2011network}
A.~El~Gamal and Y.~Kim, {\em Network information theory}.
\newblock Cambridge University Press, 2011.

\end{thebibliography}

%%%%%%%%%%%%%%%%%%%%%%%%%%%%%%%%%%%%%%%%%%%%%%%%%%%%%%%%%%%%%%%%%%%%%%%%%%%%%%%%%%%%%%%%%%%%%%%%%%%%%%%%%%%%%%%%%%%%%%%%%%%%%%%%%%%%
%\section*{Acknowledgment}
%\color{red}{The authors would like to thank}.
%%%%%%%%%%%%%%%%%%%%%%%%%%%%%%%%%%%%%%%%%%%%%%%%%%%%%%%%%%%%%%%%%%%%%%%%%%%%%%%%%%%%%%%%%%%%%%%%%%%%%%%%%%%%%%%%%%%%%%%%%%%%%%%%%%%%
\end{document}